\UseRawInputEncoding

\documentclass[11pt, reqno]{article}

\usepackage[a4paper, left=3cm, right=3cm, top=3cm]{geometry}
\usepackage{euscript,amscd,amsgen,amsfonts,amssymb,latexsym,graphicx,psfrag,multicol}
\usepackage{amsmath}
\usepackage{amsthm}
\usepackage{dsfont}
\usepackage[english]{babel}
\usepackage{abstract}
\usepackage{bbm}
\usepackage{hyperref}
\usepackage{nameref}
\usepackage{cleveref}
\usepackage{extarrows}
\usepackage[arrow, matrix, curve]{xy}
\usepackage{pdfpages}
\usepackage{braket}
\usepackage{tikz}
\usepackage{MnSymbol}
\usetikzlibrary{through}
\usepackage[numbers,sort]{natbib}
\usepackage{tabularx}

\def \bE {{\mathbb{E}}}
\def \bF {{\mathbb{F}}}
\def \bN {{\mathbb{N}}}
\def \bR {{\mathbb{R}}}

\def \cA {{\mathcal{A}}}

\def \cF {{\mathcal{F}}}
\def \cG {{\mathcal{G}}}

\def \cQ {{\mathcal{Q}}}

\def \cV {{\mathcal{V}}}

\def \cZ {{\mathcal{Z}}}

\def \bP {{\textbf{P}}}
\def \bQ {{\textbf{Q}}}

\def \d {\text{d}}

\def \CPS {{\text{CPS}}}

\def \loc {{\text{loc}}}

\def \SQ {{\tilde{S}^\bQ}}
\def \SQQ {{\tilde{S}^{\bQ_0}}}

\def \hSQ {{\widehat{S}^{\widehat{\bQ}}}}

\def \tSn {{\tilde{S}^n}}

\def \bSQ {{\bar{S}^{\bar{\bQ}}}}
\def \bbQ {{\overline{\bQ}}}
\def \bQQ {{\bQ_0}}

\renewcommand{\emptyset}{\font\cmsy = cmsy10 at 10pt
 \hbox{\cmsy \char 59}
}

\numberwithin{equation}{section}

\theoremstyle{plain}                    
\newtheorem{theorem}{Theorem}[section]
\newtheorem{lemma}[theorem]{Lemma}

\newtheorem{proposition}[theorem]{Proposition}
\newtheorem{remark}[theorem]{Remark}

\theoremstyle{definition}
\newtheorem{definition}[theorem]{Definition}

\newtheorem{example}[theorem]{Example}

\newtheorem{assumption}[theorem]{Assumption}

\theoremstyle{remark}

\newcommand*\samethanks[1][\value{footnote}]{\footnotemark[#1]}

\renewcommand\labelenumi{(\roman{enumi})}    
\renewcommand\theenumi\labelenumi                                   

\DeclareMathSymbol{\varnothing}{\mathord}{AMSb}{"3F}
\DeclareMathOperator*{\esssup}{ess\,sup}
\DeclareMathOperator*{\essinf}{ess\,inf}
\title{Asset Price Bubbles in markets with Transaction Costs}
\author{F. Biagini, T. Reitsam}
\date{\today}

\makeatletter
\def\@tocline#1#2#3#4#5#6#7{\relax
  \ifnum #1>\c@tocdepth 
  \else
    \par \addpenalty\@secpenalty\addvspace{#2}%
    \begingroup \hyphenpenalty\@M
    \@ifempty{#4}{%
      \@tempdima\csname r@tocindent\number#1\endcsname\relax
    }{%
      \@tempdima#4\relax
    }%
    \parindent\z@ \leftskip#3\relax \advance\leftskip\@tempdima\relax
    \rightskip\@pnumwidth plus4em \parfillskip-\@pnumwidth
    #5\leavevmode\hskip-\@tempdima
      \ifcase #1
       \or\or \hskip 1em \or \hskip 2em \else \hskip 3em \fi%
      #6\nobreak\relax
    \dotfill\hbox to\@pnumwidth{\@tocpagenum{#7}}\par
    \nobreak
    \endgroup
  \fi}

\newcommand*{\rom}[1]{\expandafter\@slowromancap\romannumeral #1@}
\makeatother

\author{Francesca Biagini\thanks{Workgroup Financial and Insurance Mathematics, Department of Mathematics, Ludwig-Maximilians Universit{\"a}t, Theresienstraße 39, 80333 Munich, Germany. Emails: biagini@math.lmu.de, reitsam@math.lmu.de.} \and Thomas Reitsam\samethanks[1] \thanks{The financial support of the Verein zur Versicherungswissenschaft M\"unchen e.V. is gratefully acknowledged} }
 
 \title {A dynamic version of the super-replication theorem under proportional transaction costs}

\begin{document}

\maketitle
\begin{abstract}
We extend the super-replication theorems of \cite{schachermayersuperreplication} in a dynamic setting, both in the num\'eraire-based as well as in the num\'eraire-free setting. For this purpose, we generalize the notion of admissible strategies. In particular, we obtain a well-defined super-replication price process, which is right-continuous under some regularity assumptions.
\end{abstract}\noindent
\textbf{Keywords:}\begin{small} super-replication, proportional transaction costs, consistent price systems\end{small}
\smallbreak\noindent
\smallbreak\noindent

\section{Introduction}
\label{sec:introduction}

In this paper, we provide a dynamic approach for super-replication dualities in market models with proportional transaction costs and finite time horizon $T$. \\
Since \cite{quenez}, the concept of super-replication has been thoroughly studied in a variety of market models. For the frictionless case, we refer to \cite{bartl2019duality}, \cite{bartl2020pathwise}, \cite{burzoni2017model}, \cite{carassus2007class}, \cite{quenez}, \cite{kramkovduality}, \cite{nutz2015robust}, \cite{nutz2012superhedging}, \cite{obloj2021robust}, \cite{touzi1999super}. Among other market models, super-replication dualities were established under proportional transaction costs, see e.g. \cite{campischachermayer}, \cite{cvitanic}, \cite{cvitanic1999closed}, \cite{kabanov99}, \cite{schachermayersuperreplication}, \cite{touzi1999super}. \\
The aim of this paper is to extend the results from Theorem 4.1 of \cite{campischachermayer}, and Theorems 1.4 and 1.5 of \cite{schachermayersuperreplication} to a dynamic version in order to obtain a well-defined super-replication price process in market models with transaction costs. In the frictionless case the super-replication price process is a supermartingale, see \cite{quenez}, \cite{kramkovduality}. In general, this is not the case in presence of transaction costs. In \cite{biagini2019asset}, Proposition 3.4 provides a dynamic version of the super-replication duality in the local case for a special case, where the liquidation value of the contingent claim is bounded from below by a random variable satisfying strong integrability conditions. Here, we formulate the dynamic super-replication dualities in the local and non-local case in a more general setting.\\
For this purpose, we follow the approach of \cite{campischachermayer} and \cite{schachermayersuperreplication}. In \cite{campischachermayer}, the super-replication duality for $t=0$ is proved for the num\'eraire-free scenario. In \cite{schachermayersuperreplication}, the result of \cite{campischachermayer} is formulated in a one-dimensional setting and extended to a num\'eraire-based version. The proofs of Theorem 4.1 of \cite{campischachermayer} and Theorem 1.5 of \cite{schachermayersuperreplication} strongly rely on the bipolar theorem of \cite{kabanovlast}.\\
In our approach, a fundamental role is played by the definition of admissible strategies. More precisely, we consider strategies on $[s,T]$, for $s\in [0,T]$, with random initial endowments depending on the information available at time $s$. The liquidation value of the corresponding portfolio is allowed to be be bounded from below by a random variable rather than by a constant, see Definition \ref{def:strategies} and \cite{biagini2019asset}. In this general setting, we then prove an analogous version of the bipolar theorem of \cite{kabanovlast}, see Theorem \ref{thm:bipolarext}, and the super-replication duality for the num\'eraire-free case, see Theorem \ref{thm:superreplicationdynamic}. Finally, we show that the obtained super-replication price process is well-defined. Further, we provide sufficient conditions such that the super-replication price process is right-continuous. Our analysis is motivated by the study of asset price bubbles in the presence of transaction costs, see \cite{biagini2019asset}.\\
The paper is organized as follows. In Section \ref{sec:setting}, we present the setting, define admissible strategies and provide a dual representation for consistent local price systems. In Section \ref{sec:dynamicsuperreplication}, we provide an extended version of the bipolar theorem, see Theorem \ref{thm:bipolarext}, and prove super-replication results for the dynamic setting in Theorem \ref{thm:superreplicationdynamic} and \ref{thm:superreplicationdynamiclocal}, respectively. Then, in Section \ref{sec:furtherproperties}, we elaborate further properties of the super-replication price process. In particular, we provide sufficient conditions such that the super-replication price process is right-continuous.

\section{Setting}
\label{sec:setting}

Let $T>0$ describe a finite time horizon and $(\Omega,\cF,(\cF_t)_{0\leq t \leq T},\bP)$ be a filtered probability space, where the filtration $\bF:=(\cF_t)_{0\leq t\leq T}$ satisfies the usual conditions of right-continuity and saturatedness with $\cF_0=\{\emptyset,\Omega\}$ and $\cF_T=\cF$. We consider a financial market model consisting of a risk-free asset $B$, normalized to $B\equiv 1$, and a risky asset $S$. Throughout the paper we assume that $S=(S_t)_{0\leq t\leq T}$ is an $\bF$-adapted stochastic process, with c\`adl\`ag and strictly positive paths such that $S_t\in L_+^1(\cF_t,\bP)$ for all $t\in [0,T]$. For trading the risky asset in the market model, proportional transaction costs $0<\lambda<1$ are charged, i.e., to buy one share of $S$ at time $t$ the trader has to pay $(1+\lambda)S_t$ and for selling one share of $S$ at time $t$ the trader receives $(1-\lambda)S_t$. The interval $[(1-\lambda)S_t,(1+\lambda)S_t]$ is called \emph{bid-ask-spread}. 
\begin{definition}
\label{def:cps}
For $0\leq s<t\leq T$, we call $\CPS(s,t)$ (resp. $\CPS_{\loc}(s,t)$) the family of pairs $(\bQ,\SQ)$ such that $\bQ$ is a probability measure on $\cF_t$, $\bQ\sim\bP|_{\cF_t}$, $\SQ$ is a martingale (resp. local martingale) under $\bQ$ on $[s,t]$, and
\begin{align}
\label{eq:cps}
(1-\lambda)S_u\leq \SQ_u\leq (1+\lambda)S_u,\quad\text{for }s \leq u\leq t.
\end{align}
A pair $(\bQ,\SQ)$ in $\CPS(s,t)$ (resp. $\CPS_{\loc}(s,t)$) is called \emph{consistent price system} (resp. \emph{consistent local price system}). By $\cQ(s,T)$ (resp. $\cQ_\loc(s,T)$) we denote the set of measures $\bQ$ such that there exists a pair $(\bQ,\SQ)\in\CPS(s,T)$ (resp. $(\bQ,\SQ)\in\CPS_\loc(s,T)$). Further, we write $L^p(\cF_s,\cQ):=\bigcap_{\bQ\in\cQ(s,T)} L^p(\cF_s,\bQ)$ and $L^p(\cF_s,\cQ_\loc):=\bigcap_{\bQ\in\cQ_\loc(s,T)} L^p(\cF_s,\bQ)$. By $L_+^p(\cF_s,\cQ)$ (resp. $L_+^p(\cF_s,\cQ_\loc)$) we denote the space of $[0,\infty)$-valued random variables $X\in L^p(\cF_s,\cQ)$ (resp. $X\in L^p(\cF_s,\cQ_\loc)$).
\end{definition}\noindent
It is well-known, that in the frictionless case the no-arbitrage condition no free lunch with vanishing risk (NFLVR) is equivalent to the existence of an equivalent local martingale measure, see \cite{delbaenschachermayer94}. In particular, the price process must be a semi-martingale and admit an equivalent local martingale measure. In contrast, under proportional transaction costs the existence of consistent (local) price systems guarantee the absence of arbitrage in the sense of Definition 4 of \cite{guasoniftap}. \\
A consistent (local) price system can be thought as a frictionless market with better conditions for traders, see \cite{schachermayercps}. Considering consistent price systems in the non-local or local sense corresponds in the frictionless case to the characterization of no arbitrage using true martingales or local martingales. In both cases the difference lies in the choice of admissible trading strategies. If we fix a num\'eraire, we can control the portfolio in units of the num\'eraire, and we do not allow short positions in the risky asset. Without num\'eraire, we also admit portfolios with short positions in the assets. See Chapter 5 of \cite{guasoniftap} for a more detailed discussion. For the convenience of the reader, we here state the assumptions that we use through out the paper.

\begin{assumption}
\label{assumption}
We assume that $S$ admits a consistent \emph{local} price system for every $0<\lambda'\leq \lambda$.
\end{assumption} \noindent

\begin{assumption}
\label{assumption2}
We assume that $S$ admits a consistent price system for every $0<\lambda'\leq \lambda$. 
\end{assumption}\noindent
We follow the approach of \cite{biagini2019asset} and define admissible trading strategies as follows.

\begin{definition}
Fix $s\in [0,T]$. A \emph{self-financing trading strategy} starting with initial endowment $(X_s^1,X_s^2)\in L_+^0(\cF_s,\bP)\times L_+^0(\cF_s,\bP)$ is a pair of $\bF$-predictable finite variation processes $(\varphi_t^1,\varphi_t^2)_{s\leq t\leq T}$ on $[s,T]$ such that
\begin{enumerate}
\item $\varphi_s^1=X_s^1$ and $\varphi_s^2=X_s^2$,
\item denoting by $\varphi_t^1=\varphi_s^1+\varphi_t^{1,\uparrow}-\varphi_t^{1,\downarrow}$ and $\varphi_t^2=\varphi_t^{2,\uparrow}-\varphi_t^{2,\downarrow}$, the Jordan-Hahn decomposition of $\varphi^1$ and $\varphi^2$ into the difference of two non-decreasing processes, starting at $\varphi_s^{1,\uparrow}=\varphi_s^{1,\downarrow}=\varphi_s^{2,\uparrow}=\varphi_s^{2,\downarrow}=0$, these processes satisfy
\begin{align}
\label{ali:diffform}
\d\varphi_t^1\leq (1-\lambda)S_t\d\varphi_t^{2,\downarrow}-(1+\lambda)S_t\d\varphi_t^{2,\uparrow},\quad s\leq t\leq  T.
\end{align}
\end{enumerate}
\end{definition}

\begin{definition}
\label{def:strategies}
Fix $s\in[0,T]$.
\begin{enumerate}
\item Let $X_s\in L_+^1(\cF_s,\cQ_\loc)$. A self-financing trading strategy $\varphi=(\varphi^1,\varphi^2)$ is called \emph{admissible in a num\'eraire-based} sense on $[s,T]$ starting with initial endowment $\varphi_s=(X_s,0)$ if there is  $M_s\in L_+^1(\cF_s,\cQ_\loc)$ such that the liquidation value $V_\tau^{liq}$ satisfies
\begin{align}
\label{ali:admnumbased}
V_\tau^{liq}(\varphi^1,\varphi^2):=\varphi_\tau^1+\left(\varphi_\tau^2\right)^+(1-\lambda)S_\tau-\left(\varphi_\tau^2\right)^-(1+\lambda)S_\tau\geq -M_s,
\end{align}
for all $[s,T]$-valued stopping times $\tau$.
\item Let $(X_s^1,X_s^2)\in L_+^1(\cF_s,\cQ)\times L_+^\infty(\cF_s,\cQ)$. A self-financing trading strategy $\varphi=(\varphi^1,\varphi^2)$ is called \emph{admissible in a num\'eraire-free} sense on $[s,T]$ starting with initial endowment $\varphi_s=(X_s^1,X_s^2)$ if there is $M_s:=(M_s^1,M_s^2)\in L_+^1(\cF_s,\cQ)\times L_+^\infty(\cF_s,\cQ)$ such that
\begin{align}
\label{ali:admnumfree}
V_\tau^{liq}(\varphi^1,\varphi^2):=\varphi_\tau^1+\left(\varphi_\tau^2\right)^+(1-\lambda)S_\tau-\left(\varphi_\tau^2\right)^-(1+\lambda)S_\tau\geq -M_s^1-M_s^2S_\tau,
\end{align}
for all $[s,T]$-valued stopping times $\tau$.
\end{enumerate}\noindent
We call a strategy satisfying \eqref{ali:admnumbased} or \eqref{ali:admnumfree} $M_s$-admissible in a num\'eraire-based or num\'eraire-free sense, respectively.
\end{definition} \noindent
The num\'eraire-based definition of admissibility corresponds to the setting of consistent local price systems. The portfolio is bounded from below, i.e., can be hedged in units of the num\'eraire. In particular, no short positions in the risky asset are admissible.\\
On the other side, the num\'eraire-free version of admissible trading strategy is used in the context of consistent price systems in the non-local sense. In this case no natural num\'eraire is needed and the portfolio is bounded from below by a position which depends on each of the assets. Thus, also short positions in the risky asset are admissible.

\begin{definition}
\label{def:contingentclaim}
A contingent claim $X_T=(X_T^1,X_T^2)$ is an $\cF_T$-measurable random variable in $L^0(\bR^2;\cF_T,\bP)$ which pays $X_T^1$ units of the bond and $X_T^2$ units of the risky asset at time $T$. 	
\end{definition}\noindent
Note that by Definition \ref{def:contingentclaim} a contingent claim is not assumed to be strictly positive. However, in the sequel we will require some lower bound properties depending on the setting, see \eqref{ali:admnumbased} and \eqref{ali:admnumfree}.
\begin{definition}
For $M_s:=(M_s^1,M_s^2)\in L_+^1(\cF_s,\cQ)\times L_+^\infty(\cF_s,\cQ)$ (resp. $M_s\in L_+^1(\cF_s,\cQ_\loc)$) we denote by $\cA_{s,T}^{M_s}$ (resp. $\cA_{s,T}^{M_s,\loc}$) the set of pairs $(\varphi_T^1,\varphi_T^1)\in L^0(\bR^2;\cF_T,\bP)$ of terminal values of self-financing trading strategies $\varphi$, starting at $\varphi_s=(0,0)$, which are $M_s$-admissible in the num\'eraire-free sense (resp. num\'eraire-based sense). Further, we denote
\begin{align}
\label{ali:defterminalvalues}
\cA_{s,T}:=\left\{\varphi_T:\varphi_T\in\cA_{s,T}^{M_s} \text{ for some } M_s=(M_s^1,M_s^2)\in L_+^1(\cF_s,\cQ)\times L_+^\infty(\cF_s,\cQ)\right\}.
\end{align}
\end{definition}\noindent
We now introduce a dual theory for consistent (local) price systems, see also \cite{campischachermayer}, \cite{guasoniftap}, \cite{kabanovlast}. For fixed $\lambda>0$ we denote by $K_t$ the solvency cone at time $t$, defined as
\begin{align}
\label{ali:solvencycone}
K_t(\omega)=\text{cone}\left\{(1+\lambda)S_t(\omega) e_1-e_2,-e_1+\frac{1}{(1-\lambda)S_t(\omega)} e_2\right\},
\end{align}
where $e_1=(1,0), e_2=(0,1)$ are the unit vectors in $\bR^2$, and by $K_t^*=(-K_t)^\circ$ the corresponding polar cone, given by
\begin{align}
\label{ali:polarcone}
\begin{split}
K_t^*(\omega)=(-K_t)^\circ(\omega)&=\left\{(y_1,y_2)\in\bR_+^2: (1-\lambda)S_t(\omega)\leq \frac{y_2}{y_1}\leq (1+\lambda)S_t(\omega)\right\}\cup\{0\}\\
&=\left\{y\in\bR^2: \langle x,y\rangle\leq 0, \forall x\in\left( -K_t(\omega)\right)\right\}\\
&=\left\{y\in\bR^2:\langle x,y\rangle \geq 0,\ \forall x\in K_t(\omega)\right\}.
\end{split}
\end{align}

\begin{definition}
\label{def:cpspolar}
Let $s\in [0,T]$. We define $\mathcal{Z}(s,T)$ (resp. $\mathcal{Z}_{\loc}(s,T)$) as the set of processes $Z=(Z_t^1,Z_t^2)_{s\leq t\leq T}$ such that $Z^1$ is a $\bP$-martingale and $Z^2$ is a $\bP$-martingale (resp. local $\bP$-martingale) and such that $Z_t\in K_t^*\backslash\{0\}$ a.s. for all $t\in[s,T]$.
\end{definition} \noindent
The following proposition from \cite{guasoniftap} provides a useful representation of consistent (local) price systems by elements in $\mathcal{Z}$ (resp. $\mathcal{Z}_\loc$) and follows directly from the definition of $K_t^*$ in \eqref{ali:polarcone}.

\begin{proposition}[Proposition 3, \cite{guasoniftap}]
\label{prop:pmartingale}
Let $s\in [0,T]$ and $Z=(Z_t^1,Z_t^2)_{s\leq t\leq T}$ be a $2$-dimensional stochastic process with $Z_T^1\in L^1(\cF_T,\bP)$. Define the measure $\bQ(Z)\ll\bP$ by $\d\bQ(Z)/\d\bP:=Z_T^1/\bE[Z_T^1]$. Then $Z\in \cZ(s,T)$ (resp. $Z\in \cZ_{\loc}(s,T)$) if and only if $(\bQ(Z),(Z^2/Z^1))$ is a consistent price system (resp. consistent local price system) on $[s,T]$.
\end{proposition} \noindent

\section{Dynamic super-replication}
\label{sec:dynamicsuperreplication}

\subsection{A Bipolar Theorem in a dynamic setting}
\label{sec:bipolar}
In this section we aim to provide a bipolar theorem which will then be used to apply techniques from \cite{campischachermayer} and \cite{schachermayersuperreplication} in the proof of the super-replication theorems. Theorem \ref{thm:bipolarext} extends the bipolar theorem of \cite{kabanovlast}, see also \cite{kabanovbook}.

\begin{definition}
\label{def:directedupwards}
	We define the partial order $\succeq$ on $L^0( \bR^2;\cF_T,\bP)$ by letting $\varphi\succeq \psi$ if and only if $V_T^{liq}(\varphi^1-\psi^1,\varphi^2-\psi^2)\geq 0$, i.e., if the portfolio $\varphi-\psi$ can be liquidated to the zero portfolio.\\
	We say that a set $\Phi\subset L^0(\bR^2;\cF_T,\bP)$ is directed upwards if for $\psi_1,\psi_2\in \Phi$ there exists $\psi\in \Phi$ with $\psi\geq \psi_1\vee \psi_2$.
\end{definition}\noindent
We denote by $L_{1,\infty}^0$ the cone in $L^0( \bR^2;\cF_T,\bP)$ given by the random variables $\xi=(\xi^1,\xi^2)$ such that $(\xi^1,\xi^2)\succeq (-M_s^1,-M_s^2)$ for some $(M_s^1,M_s^2)\in L_+^1(\cF_s,\cQ)\times L_+^\infty(\cF_s,\cQ)$. Further, we denote by $L_b^0\subset L^0( \bR^2;\cF_T,\bP)$ the cone formed by the random variables $\xi$ such that $(\xi^1,\xi^2)\succeq (-M,-M)$ for some $M>0$, following \cite{kabanovbook}.

\begin{remark}
Note that the conditional expectation
\begin{equation}
\label{eq:condexp}
\bE_\bP\left[\xi\cdot Z_T\mid\cF_s\right],\quad Z_T\in L^1(K_T^*),\ \xi\in  L_{1,\infty}^0
\end{equation}
is well-defined. In fact, by the definition of $L_{1,\infty}^0$ there exists $M_s=(M_s^1,M_s^2)\in L^1(\bR_+;\cF_s,\cQ)\times L^\infty(\bR_+;\cF_s,\cQ)$ such that $\left(\xi + M_s\right)\in K_T$ and hence $\left(\xi + M_s\right)\cdot Z_T\geq 0$. In particular, we get
\begin{align*}
\bE_\bP\left[\xi\cdot Z_T\mid\cF_s\right]=\bE_\bP\left[\left(\xi + M_s\right)\cdot Z_T - M_s\cdot Z_T\mid\cF_s\right].
\end{align*}
For non-negative random variables the conditional expectation is always well-defined, although it might be infinity. Thus, 
\[
\bE_\bP\left[\left(\xi + M_s\right)\cdot Z_T\mid\cF_s\right]
\]
is well-defined. Furthermore, we need
\begin{equation}
\label{eq:condexp2}
\bE_\bP\left[M_s\cdot Z_T\mid\cF_s\right]< \infty.
\end{equation}
First, note that $M_s\cdot Z_T\geq 0$ and hence \eqref{eq:condexp2} is well-defined. Following Section 27 of \cite{loeve1978probability}, we use that $M_s$ is $\cF_s$-measurable and $Z_T\in L^1(K_T^*,\cF_T,\bP)$ to conclude
\[
\bE_\bP\left[M_s\cdot Z_T\mid\cF_s\right]=M_s\bE_\bP\left[Z_T\mid\cF_s\right]<\infty.
\]
Therefore, 
\[
\bE_\bP\left[\xi\cdot Z_T\mid\cF_s\right]\geq -M_s\bE_\bP\left[Z_T\mid\cF_s\right]>-\infty
\]
is well-defined. However, the expectation
\[
\bE_\bP\left[M_s\cdot Z_T\right]
\]
is in general not well-defined. In contrast, for $\eta\in L_b^0$ and $Z_T \in L^1(K_T^*;\cF_T,\bP)$ as in the bipolar theorem of \cite{kabanovlast} there exists $M>0$ such that
\[
\bE_\bP\left[\eta\cdot Z_T\right]\geq -M\bE_\bP\left[Z_T\right]>-\infty
\]
is well-defined.
\end{remark}\noindent
We now extend the definition of Fatou convergence of \cite{campischachermayer}, \cite{kabanovbook}, \cite{schachermayersuperreplication}.

\begin{definition}
\label{def:fatouconvergence}
	Consider a sequence $(X_n)_{n\in\bN}=(X_n^1,X_n^2)_{n\in\bN}\subset L^0( \bR^2;\cF_T,\bP)$. We say that $(X_n)_{n\in\bN}$ is $L^0(\cF_s)$-Fatou converging to $X=(X^1,X^2)$ if $X_n\xrightarrow{a.s.} X$ and $(X_n^1,X_n^2)\succeq (-M_s^1,-M_s^2)$ for all $n\in\bN$ and some $(M_s^1,M_s^2)\in L_+^1(\cF_s,\cQ)\times L_+^\infty(\cF_s,\cQ)$.
\end{definition}\noindent
If $(-M_s^1,-M_s^2)=(-M,-M)$ for some $M\in \bR_+$, $L^0(\cF_s)$-Fatou convergence coincides with the Fatou convergence\footnote{Following \cite{kabanovlast}, \cite{schachermayersuperreplication}, let $(X_n)_{n\in\bN}\subset L^0( \bR^2;\cF_T,\bP)$. We say that $(X_n)_{n\in\bN}$ is Fatou converging to $X=(X^1,X^2)$ if $X_n\xrightarrow{a.s.}X$ and $(X_n^1,X_n^2)\succeq (-M,-M)$ for all $n\in\bN$ and some $M>0$.} as defined in \cite{campischachermayer}, \cite{kabanovbook}, \cite{schachermayersuperreplication}.
\begin{remark}
	\label{rem:resultsfromschachermayer}
	Note that Lemma 3.1 of \cite{schachermayersuperreplication} also holds in our setting. We omit the proof since it holds in our setting with minor modifications. The lemma ensures that the total variation of strategies, which are $M_s$-admissible in the num\'eraire-free sense, remains bounded in $L^0(\cF_T,\bP)$. Furthermore, this implies that Theorem 3.4 and Theorem 3.6 of \cite{schachermayersuperreplication} hold true as well. Theorem 3.4 (resp. Theorem 3.6) of \cite{schachermayersuperreplication} guarantee that $A_{s,T}^{M_s}\subset L^0(\cF_s,\bP)$ (resp. $A_{s,T}^{M_s,\loc}\subset L^0(\cF_s,\bP)$) is closed with respect to the topology of convergence in measure. For more details we refer to \cite{thesis}.
\end{remark}

\begin{theorem}
\label{thm:bipolarext}
	Let $s\in[0,T]$. It holds that
	\begin{align}
	\label{ali:bipolarextension}
	 \cA_{s,T}=\left\{\xi\in L_{1,\infty}^0:\bE_\bP\left[\xi\cdot Z_T\mid\cF_s\right]\leq\esssup\limits_{\eta\in \cA_{s,T}}\bE_\bP\left[\eta\cdot Z_T\mid\cF_s\right]\ \forall Z_T \in L^1(K_T^*;\cF_T,\bP)\right\},
	\end{align}
	where $A_{s,T}$ is defined in \eqref{ali:defterminalvalues}.
\end{theorem}

\begin{proof}
	The inclusion $``\subseteq"$ is trivial.\\
For the reverse inclusion we make use of the bipolar theorem of \cite{kabanovlast}, Theorem 4.2, see also Theorem 5.5.3 of \cite{kabanovbook}. Suppose the conditions of Theorem 4.2 of \cite{kabanovlast} are satisfied for $\cA_{s,T}\cap L_b^0$. Then we obtain
	\begin{align}
		\cA_{s,T}\cap L_b^0=\left\{\xi\in L_b^0:\bE_\bP\left[\xi \cdot Z_T\right]\leq \sup_{\eta\in \cA_{s,T}\cap L_b^0}\bE_\bP\left[\eta \cdot Z_T\right]\ \forall Z_T\in L^1(K_T^*;\cF_T,\bP)\right\}.
	\end{align}
	First, we show
	\begin{align}
	\label{ali:supremumequation}
	\bE_\bP\left[\esssup_{\eta\in \cA_{s,T}\cap L_b^0}\bE_\bP\left[\eta\cdot  Z_T\mid\cF_s\right]\right]=\sup_{\eta\in \cA_{s,T}\cap L_b^0}\bE_\bP\left[\eta \cdot Z_T\right].
	\end{align}
  	By monotonicity we have that
	\begin{align}
		\bE_\bP\left[\esssup_{\eta\in \cA_{s,T}\cap L_b^0}\bE_\bP\left[\eta \cdot Z_T\mid\cF_s\right]\right]\geq \sup_{\eta\in \cA_{s,T}\cap L_b^0}\bE_\bP\left[\eta\cdot  Z_T\right].
	\end{align}
	For $Z_T\in L^1(K_T^*;\cF_T,\bP)$ we define $\Phi_{Z_T}:=\{\bE_\bP[\eta Z_T\mid\cF_s]: \eta \in L_b^0\}$. It is easy to see that $\Phi_{Z_T}$ is directed upwards, since for $\eta, \tilde{\eta}\in L_b^0$ we get for 
	\begin{align}
		D_s:=\left\{\bE_\bP\left[\eta \cdot Z_T\mid\cF_s\right]\geq \bE_\bP\left[\tilde{\eta}\cdot Z_T\mid\cF_s\right]\right\}\in \cF_s
	\end{align}
	that $\psi:=\eta \mathds{1}_{D_s}+\tilde{\eta}\mathds{1}_{D_s^c}\in L_b^0$ and 
	\begin{align}
		\bE_\bP\left[\psi\cdot  Z_T\mid\cF_s\right]\geq \left(\bE_\bP\left[\eta\cdot  Z_T\mid\cF_s\right]\vee \bE_\bP\left[\tilde{\eta}\cdot Z_T\mid\cF_s\right]\right).
	\end{align}
	By Theorem A.33 of \cite{follmerschied} there exists $(\eta_n)_{n\in\bN}\subset L_b^0$ such that 
	\begin{align}
	\bE_\bP\left[\eta_n \cdot Z_T\mid\cF_s\right] \uparrow \esssup_{\eta\in \cA_{s,T}\cap L_b^0}\bE_\bP\left[\eta\cdot  Z_T\mid\cF_s\right],\text{ as }n\to \infty.
	\end{align}
	Hence, we get
	\begin{align}
		\bE_\bP\left[\esssup_{\eta\in \cA_{s,T}\cap L_b^0}\bE_\bP\left[\eta\cdot  Z_T\mid\cF_s\right]\right]=\bE_\bP\left[\lim_{n\to\infty}\bE_\bP\left[\eta_n \cdot Z_T\mid\cF_s\right]\right].
	\end{align}
	Note that $0\in \cA_{s,T}\cap L_b^0$ and thus $\esssup_{\eta\in \cA_{s,T}\cap L_b^0}\bE_\bP[\eta Z_T\mid\cF_s]\geq 0$. In particular we can assume without loss of generality that $\bE_\bP[\eta_n\cdot  Z_T\mid\cF_s]\geq 0$ a.s. for all $n\in\bN$. Thus we obtain by monotone convergence
	\begin{align}
		\bE_\bP\left[\lim_{n\to\infty}\bE_\bP\left[\eta_n \cdot Z_T\mid\cF_s\right]\right]=\lim_{n\to\infty}\bE_\bP\left[\eta_n \cdot Z_T\right]\leq \sup_{\eta\in \cA_{s,T}\cap L_b^0}\bE_\bP\left[\eta \cdot Z_T\right],
	\end{align}
	and thus \eqref{ali:supremumequation} is fulfilled. Therefore, we have that
	\begin{align}
		&\left\{\xi\in L_b^0:\bE_\bP\left[\xi \cdot Z_T\mid\cF_s\right]\leq \esssup_{\eta\in \cA_{s,T}\cap L_b^0}\bE_\bP\left[\eta\cdot  Z_T\mid\cF_s\right],\ \forall Z_T\in L^1(K_T^*;\cF_T,\bP)\right\}\\
		\subseteq &\left\{\xi\in L_b^0:\bE_\bP\left[\xi \cdot Z_T\right]\leq \sup_{\eta\in \cA_{s,T}\cap L_b^0}\bE_\bP\left[\eta \cdot Z_T\right],\ \forall Z_T \in L^1(K_T^*;\cF_T,\bP)\right\}.
	\end{align}
	In particular, under the assumption that the conditions for Theorem 4.2 of \cite{kabanovlast} are fulfilled we get
	\begin{align}
	\label{ali:eqforclosure}
		I:=\left\{\xi\in L_b^0:\bE_\bP\left[\xi \cdot Z_T\mid\cF_s\right]\leq \esssup_{\eta\in \cA_{s,T}\cap L_b^0}\bE_\bP\left[\eta \cdot Z_T\mid\cF_s\right]\ \forall Z_T\in L^1(K_T^*;\cF_T,\bP)\right\}\subseteq \cA_{s,T} \cap L_b^0.
	\end{align}
We now prove that the assumptions of Theorem 4.2 of \cite{kabanovlast} are indeed fulfilled for $\cA_{s,T}\cap L_b^0$. The arguments are similar to the proof of Theorem 4.1 of \cite{campischachermayer} and the proof of Theorem 1.5 of \cite{schachermayersuperreplication}. First, Theorem 3.6 of \cite{schachermayersuperreplication} and Remark \ref{rem:resultsfromschachermayer} imply directly that $\cA_{s,T}\cap L_b^0$ is Fatou-closed. By standard arguments $\cA_{s,T}\cap L^\infty(\cF_T,\bP)$ is dense in $\cA_{s,T}\cap L_b^0$ with respect to Fatou-convergence. In fact, let $\varphi=(\varphi^1,\varphi^2)\in \cA_{s,T}\cap L_b^0$, i.e. $\varphi\succeq (-M,-M)$ for some $M>0$. We define the sequence $\varphi_n:=\varphi \mathds{1}_{\{|\varphi|\leq n\}}-(M,M)\mathds{1}_{\{|\varphi|>n\}}$. Then $(\varphi_n)_{n\in\bN}\subset (\cA_{s,T}\cap L_b^0)\cap L^\infty(\cF_T,\bP)$ and $\varphi_n\xrightarrow{a.s.} \varphi$. Furthermore, $\varphi_n\succeq (-M,-M)$ for all $n\in\bN$ which guarantees that $\varphi_n$ Fatou-converges to $\varphi$ and that $(\cA_{s,T}\cap L_b^0)\cap L^\infty(\cF_T,\bP)$ is dense with respect to Fatou-convergence in $\cA_{s,T}\cap L_b^0$. From this construction we also observe that $-L^\infty(K_T;\cF_T,\bP)\subset \cA_{s,T}\cap L_b^0$. Therefore, the conditions of Theorem 4.2 of \cite{kabanovlast} are satisfied for $\cA_{s,T}\cap L_b^0$.\\
	It is left to show that $L_b^0$ is dense in $L_{1,\infty}^0$ with respect to $L^0(\cF_s)$-Fatou-convergence. To see this, let $\xi \in L_{1,\infty}^0$. Then $\xi\succeq (-M_s^1,-M_s^2)$ for some $(M_s^1,M_s^2)\in L_+^1(\cF_s,\cQ)\times L_+^\infty(\cF_s,\cQ)$. If we set $\xi_n:=\xi^+-(\xi^-\wedge n)$, then $\xi_n\in L_b^0$ for all $n\in\bN$, $\xi_n\succeq (-M_s^1,-M_s^2)$ and $\xi_n\xrightarrow{a.s.} \xi$. Furthermore, $\cA_{s,T}$ is $L^0(\cF_s)$-Fatou closed by Theorem 3.6 and Remark \ref{rem:resultsfromschachermayer} and hence
	\begin{align*}
		\esssup_{\eta\in \cA_{s,T}\cap L_b^0}\bE_\bP\left[\eta Z_T\mid\cF_s\right]=\esssup_{\eta\in \cA_{s,T}}\bE_\bP\left[\eta Z_T\mid\cF_s\right],
	\end{align*}
	and
	\begin{align*}
		\bar{I}=\left\{\xi\in L_{1,\infty}^0:\bE_\bP\left[\xi Z_T\mid\cF_s\right]\leq \esssup_{\eta\in \cA_{s,T}\cap L_b^0}\bE_\bP\left[\eta Z_T\mid\cF_s\right]\ \forall Z_T\in L^1(K_T^*;\cF_T,\bP)\right\},
	\end{align*}
	where the closure is taken with respect to $L^0(\cF_s)$-Fatou convergence. Then
	\begin{align*}
		\bar{I}\subset \overline{\cA_{s,T}\cap L_b^0}=\cA_{s,T}.
	\end{align*}	
	This concludes the proof.
\end{proof}\noindent

\subsection{Super-replication theorems in a dynamic setting}
\label{sec:superreplication}

In this section we prove dynamic super-replication results in the context of local and non-local consistent price systems, respectively. We start with the non-local version.

\begin{theorem}
\label{thm:superreplicationdynamic}
Let Assumption \ref{assumption2} hold and $s\in[0,T]$. Let $X_T=(X_T^1,X_T^2)$ be a contingent claim such that 
\begin{align}
\label{ali:superreplicationdynamicconditionclaim}
X_T^1+(X_T^2)^+(1-\lambda)S_T-(X_T^2)^-(1+\lambda)S_T\geq -M_s^1-M_s^2S_T,
\end{align}
for some $(M_s^1,M_s^2)\in L_+^1(\cF_s,\cQ)\times L_+^\infty(\cF_s,\cQ)$. For a random variable $X_s=(X_s^1,X_s^2)\in L_+^1(\cF_s,\cQ)\times L_+^\infty(\cF_s,\cQ)$ the following assertions are equivalent:
\begin{enumerate}
\item There is a self-financing trading strategy $\varphi=(\varphi_t^1,\varphi_t^2)_{s\leq t\leq T}$ with $\varphi_s=(X_s^1,X_s^2)$ and $\varphi_T=(X_T^1,X_T^2)$ which is admissible in a num\'eraire-free sense on the interval $[s,T]$, see \eqref{ali:admnumfree}.
\item For every consistent price system $(\bQ,\SQ)\in \CPS(s,T)$ we have
\begin{align}
\bE_\bQ\left[X_T^1-X_s^1+(X_T^2-X_s^2)\SQ_T\mid\cF_s\right]\leq 0.
\end{align}      
\end{enumerate}
\end{theorem}\noindent

\begin{proof}[Proof of Theorem \ref{thm:superreplicationdynamic}]
	$``i)\Rightarrow ii)"$ Let $\varphi=(\varphi_t^1,\varphi_t^2)_{s\leq t\leq T}$ be an admissible strategy in the num\'eraire-free sense such that $\varphi_s=(X_s^1,X_s^2)$ and $\varphi_T=(X_T^1,X_T^2)$. By Proposition 3 of \cite{schachermayeradmissible} and Remark 2.8 of \cite{biagini2019asset}, $(\varphi_t^1+\varphi_t^2\SQ_t)_{s\leq t\leq T}$ is an optional strong supermartingale for all $(\bQ,\SQ)\in\CPS(s,T)$. In particular, we obtain
	\begin{align}\bE_\bQ\left[X_T^1+X_T^2\SQ_T\mid\cF_s\right]=\bE_\bQ\left[\varphi_T^1+\varphi_T^2\SQ_T\mid\cF_s\right]
	\leq \varphi_s^1+\varphi_s^2\SQ_s=X_s^1+X_s^2\SQ_s,
	\end{align}
	for all $(\bQ,\SQ)\in\CPS(s,T)$.\\
	$``ii)\Rightarrow i)"$ Note that for any contingent claim $X_T$, there is an admissible strategy in the num\'eraire-free sense $\varphi$ with $\varphi_s=(X_s^1,X_s^2)$ and $\varphi_T=(X_T^1,X_T^2)$ if and only if $\tilde{X}_T:=(X_T^1-X_s^1,X_T^2-X_s^2)\in\cA_{s,T}$, i.e., if there is an admissible strategy in the num\'eraire-free sense $\tilde{\varphi}$ with $\tilde{\varphi}_s=(0,0)$ and $\tilde{\varphi}_T=(\tilde{X}_T^1,\tilde{X_T^2})$. Thus, it is enough to show that for $\tilde{X}_T\notin\cA_{s,T}$ there exists a consistent price system in the non-local sense $(\bQ,\SQ)\in\CPS(s,T)$ such that 
	\begin{align}
		\bP(\bE_\bQ[\tilde{X}_T^1+\tilde{X}_T^2\SQ_T\mid\cF_s]>0)>0.
	\end{align}
Let $\tilde X_T\notin\cA_{s,T}$, then by Theorem \ref{thm:bipolarext} there exists $Y=(Y^1,Y^2)\in L^1(K_T^*;\cF_T,\bP)$ such that
	 \begin{align}
		\label{ali:notina}
		\bP\left(B^Y\right)>0,
	\end{align}
	where
	\begin{align}
	\label{ali:notina2}
		B^Y:=\left\{\omega\in\Omega:\ \bE_\bP\left[Y^1\tilde X_T^1+Y^2\tilde X_T^2\mid\cF_s\right](\omega)>\esssup_{\eta\in\cA_{s,T}}\bE_\bP\left[\eta^1Y^1+\eta^2 Y^2\mid\cF_s\right](\omega)\right\}\in\cF_s.
	\end{align}
We now construct a consistent price system $\hat{Z}=(\hat{Z}_t^1,\hat{Z}_t^2)_{s\leq t\leq T}$ as follows. By Proposition \ref{prop:pmartingale} we can represent any consistent price system in the non-local sense $(\bQ,\SQ)$ by a pair $(Z^1,Z^2)$ such that $Z^i$ is a $\bP$-martingale for $i=1,2$ and $Z_t^2/Z_t^1$ takes values in the bid-ask spread, i.e., $(Z_t^1,Z_t^2)\in K_t^*\backslash\{0\}$ almost surely. Conversely, if for a process $(Z^1,Z^2)$ we have that $(Z_t^1,Z_t^2)\in K_t^*\backslash \{0\}$ and $Z^i$ is a $\bP$-martingale for $i=1,2$, then $\d\bQ(Z)/d\bP:=Z_T^1/\bE_\bP[Z_T^1]$ and $\SQ(Z):=Z^2/Z^1$ defines a consistent price system in the non-local sense.\\
We start by defining the process $Z=(Z_t^1,Z_t^2)_{s\leq t\leq T}$ by $Z_t^i:=\bE_\bP\left[Y^i\mathds{1}_{B^Y}\mid\cF_t\right]$, $s\leq t\leq T$, $i=1,2$. Note that $0\in\cA_{s,T}$. Therefore we obtain by \eqref{ali:notina} and \eqref{ali:notina2}, that 
\[
\bP(\bE_\bP[Z_T^1\tilde X_T^1+Z_T^2\tilde X_T^2\mid\cF_s]>0)\geq \bP(B^Y)>0.
\]
We show that $Z_t\in K_t^*$ a.s. for $t\in [s,T]$. In particular, if $Z$ was $\bR_+\backslash\{0\}$-valued it is a consistent price system in the non-local sense.\\
Consider the process 
	\begin{align}
	\label{ali:processforconstructingcps}
		\psi_u:=-\nu\gamma\mathds{1}_{[t,T]}(u),\quad u\in[s,T],
	\end{align}
	for some $t\in [s,T]$ and arbitrary random variables $\nu\in L^\infty(\bR_+;\cF_t,\bP)$ and $\gamma\in L^\infty(K_T;\cF_t,\bP)$. As $\psi$ is almost surely bounded, $\psi$ is an admissible strategy in the num\'eraire-free sense and
	\begin{align}
	\label{ali:equalityforconstructioncps}
		\bE_\bP\left[Y\cdot \psi_T\mathds{1}_{B^Y}\mid\cF_s\right]=\bE_\bP\left[Z_T\cdot\psi_T\mid\cF_s\right]=-\bE_\bP\left[\nu\gamma\cdot Z_t\mid\cF_s\right].
	\end{align} 
Because $\psi\in \cA_{s,T}$ \eqref{ali:notina2} and \eqref{ali:equalityforconstructioncps} imply for $\omega\in B^Y$ that
	\[
	\bE_\bP\left[\nu\gamma Z_t\mid\cF_s\right](\omega)>-\bE_\bP\left[Y \tilde X_T\mathds{1}_{B^Y}\mid \cF_s\right](\omega)=-ä\left(\bE_\bP\left[Y\cdot \tilde X_T\mid\cF_s\right]\mathds{1}_{B^Y}\right)(\omega),
	\]
where we used that $B^Y\in \cF_s$. For $\omega\in (B^Y)^c$ we have
\[
\bE_\bP\left[\nu\gamma Z_t\mid\cF_s\right](\omega)= -\left(\bE_\bP\left[Y\cdot \tilde X_T\mid\cF_s\right]\mathds{1}_{B^Y}\right)(\omega)=0.
\]
Because $\nu$ is arbitrary, we can deduce that $Z_t\gamma\geq 0$ for all $\gamma\in L^\infty(K_t;\cF_t,\bP)$, yielding $Z_t\in K_t^*$ a.s.. Thus $Z$ is a $\bP$-martingale satisfying $Z_t\in K_t^*$ for all $t\in[s,T]$. 	\\
It is still possible that $\bP(Z_t=0)>0$ for some $t\in [s,T]$ and thus $Z$ is not necessarily a consistent price systems. We now construct the desired consistent price system $\hat Z=(\hat Z_t^1,\hat Z_t^2)_{t\in [s,T]}$ as follows. Take any consistent price system in the non-local sense $\tilde{Z}=(\tilde{Z}_t^1,\tilde{Z}_t^2)_{s\leq t \leq T}$. Then for a suitable\footnote{For instance:\begin{align*}
		\beta(\omega):=\begin{cases}
			1,&\quad \omega \in (B^Y)^c,\\
			\frac{\bE_\bP\left[Z_T\cdot\tilde X_T\mid\cF_s\right]}{\left|\bE_\bP\left[\tilde{Z}_T\cdot X_T\mid\cF_s\right]\right|+\bE_\bP\left[Z_T\cdot \tilde X_T\mid\cF_s\right]},&\quad  \omega\in B^Y.
		\end{cases}
	\end{align*}} $\beta\in L^\infty((0,1);\cF_s,\bP)$ we define $\hat{Z}_t^i:=(\beta \tilde{Z}_t^i+(1-\beta)Z_t^i)$, $i=1,2$. Then $\hat Z_t\in K_t^*\backslash\{0\}$, $t\in [s,T]$, and
	\begin{align}
		\bP\left(\bE_\bP\left[\hat{Z}_T^1\tilde X_T^1+\hat{Z}_T^2\tilde X_T^2\mid\cF_s\right]>0\right)>0.
	\end{align}
	Clearly, $\hat{Z}$ is strictly positive and hence a consistent price system in the non-local sense on $[s,T]$.
	For $\tilde X_T\notin \cA_{s,T}$ we have constructed a consistent price system in the non-local sense $\hat{Z}=(\hat{Z}_t^1,\hat{Z}_t^2)_{s\leq t\leq T}$ satisfying $\bP(\bE_\bP[\hat{Z}_T^1\tilde X_T^1+\hat{Z}_T^2\tilde X_T^2\mid\cF_s]>0)>0$. This concludes the proof.
\end{proof}\noindent
Now we can also prove the local or num\'eraire-based version of the super-replication theorem.

\begin{theorem}
\label{thm:superreplicationdynamiclocal}
Let Assumption \ref{assumption} hold and $s\in[0,T]$. Let $X_T=(X_T^1,X_T^2)$ be a contingent claim such that 
\begin{align}
\label{ali:superreplicationdynamiclocaladm}
X_T^1+(X_T^2)^+(1-\lambda)S_T-(X_T^2)^-(1+\lambda)S_T\geq -M_s
\end{align}
for some $M_s\in L_+^1(\cF_s,\cQ_\loc)$. For a random variable $X_s=(X_s^1,0)\in L^1(\cF_s,\cQ_\loc)$ the following assertions are equivalent:
\begin{enumerate}
\item\label{item:superreploc1} There is a self-financing trading strategy $\varphi=(\varphi_t^1,\varphi_t^2)_{s\leq t\leq T}$ with $\varphi_s=(X_s^1,0)$ and $\varphi_T=(X_T^1,X_T^2)$ which is admissible in a num\'eraire-based sense on the interval $[s,T]$, see \eqref{ali:admnumbased}.
\item\label{item:superreploc2} For every consistent price system $(\bQ,\SQ)\in \CPS_\loc(s,T)$ we have
\begin{align}
\label{ali:superreplicationdynamicalbound}
\bE_\bQ\left[X_T^1-X_s^1+X_T^2\SQ_T\mid\cF_s\right]\leq 0.
\end{align}
\end{enumerate}
\end{theorem}

\begin{proof}
 The arguments are identical to the proof Theorem 1.4 of \cite{schachermayersuperreplication}. Note that Theorem 3.4 of \cite{schachermayersuperreplication} as well as Theorem 2 of \cite{schachermayeradmissible} are also valid for our setting. For the proof and further details we refer to \cite{thesis}.
\end{proof}\noindent
These duality results can also be formulated in the following way.

\begin{proposition}
\label{prop:dualitylocal}
Let Assumption \ref{assumption} hold and $s\in [0,T]$. Let $X_T=(X_T^1,X_T^2)$ be a contingent claim such that 
\begin{align}
\label{ali:superreplicationdynamiclocaladmprop}
X_T^1+(X_T^2)^+(1-\lambda)S_T-(X_T^2)^-(1+\lambda)S_T\geq -M_s
\end{align}
for some $M_s\in L_+^1(\cF_s,\cQ_\loc)$. If
\begin{align*}
	\esssup_{(\bQ,\SQ)\in\CPS_\loc(s,T)}\bE_\bQ\left[X_T^1+X_T^2\SQ_T\mid\cF_s\right]\in L^1(\cF_s,\cQ_\loc),
\end{align*}
then we have 
\begin{align}
\label{ali:propduality}
\begin{split}
	&\essinf\left\{\xi_s^1\in L_+^1(\cF_s,\cQ_\loc): \exists\varphi\in \cV_{s,T}^\loc(\xi_s,\lambda)\text{ with }\varphi_s=(\xi_s,0),\ \varphi_T=X_T\right\}\\
		=&\esssup_{(\bQ,\SQ)\in\CPS_\loc(s,T)}\bE_\bQ\left[X_T^1+X_T^2\SQ_T\mid\cF_s\right].
\end{split}
\end{align}
\end{proposition} \noindent

\begin{proof}
Since $X_T$ satisfies the conditions of Theorem \ref{thm:superreplicationdynamiclocal}, we have
\begin{align}
\label{ali:proofduality1}
\begin{split}
&\left\{X_s\in L^1(\cF_s,\cQ_\loc):\ \exists \varphi \in\cV_{s,T}^\loc(X_s,\lambda) \text{ with }\varphi_s=(X_s,0)\text{ and }\varphi_T=(X_T^1,X_T^2)\right\}\\
=&\underbrace{\left\{X_s\in L^1(\cF_s,\cQ_\loc):\ \bE_\bQ\left[X_T^1+X_T^2\SQ_T\mid\cF_s\right]\leq X_s,\ \forall(\bQ,\SQ)\in \CPS_\loc(s,T)\right\}}_{=:D_s}
\end{split}
\end{align}
It is left to show that
\begin{align}
\essinf D_s=\esssup_{(\bQ,\SQ)\in\CPS_\loc(s,T)}\bE_\bQ\left[X_T^1+X_T^2\SQ_T\mid\cF_s\right].
\end{align}
For the first direction ``$\leq$'' we get that $\esssup_{(\bQ,\SQ)\in\CPS_\loc(s,T)}\bE_\bQ[X_T^1+X_T^2\SQ_T\mid\cF_s]\in D_s$, because $\esssup_{(\bQ,\SQ)\in\CPS_\loc(s,T)}\bE_\bQ[X_T^1+X_T^2\SQ_T\mid\cF_s]\in L^1(\cF_s,\cQ_\loc)$. \\
For the reverse direction ``$\geq$'' we have that $\essinf D_s\geq \bE_\bQ[X_T^1+X_T^2\SQ_T\mid\cF_s]$ for all $(\bQ,\SQ)\in\CPS_\loc(s,T),$ which implies by the definition of the essential supremum that $$\essinf D_s\geq \esssup_{(\bQ,\SQ)\in\CPS_\loc(s,T)}\bE_\bQ[X_T^1+X_T^2\SQ_T\mid\cF_s].$$
\end{proof} \noindent

\begin{proposition}
\label{prop:dualitynonlocal}
Let Assumption \ref{assumption2} hold.  Let $X_T=(X_T^1,X_T^2)$ be a contingent claim such that 
\begin{align}
\label{ali:superreplicationdynamiclocaladm1}
X_T^1+(X_T^2)^+(1-\lambda)S_T-(X_T^2)^-(1+\lambda)S_T\geq -M_s^1-M_s^2S_\tau
\end{align}
for some $(M_s^1,M_s^2)\in L_+^1(\cF_s,\cQ)\times L_+^\infty(\cF_s,\cQ)$. If
\begin{align*}
	\esssup_{(\bQ,\SQ)\in\CPS_\loc(s,T)}\bE_\bQ\left[X_T^1+X_T^2\SQ_T\mid\cF_s\right]\in L^1(\cF_s,\cQ),
\end{align*}
then we have 
\begin{align}
\label{ali:propdualitynonlocal}
\begin{split}
	&\essinf\left\{\xi_s\in L_+(\cF_s,\cQ): \exists\varphi\in \cV_{s,T}(\xi,\lambda)\text{ with }\varphi_s=(\xi_s,0)\ \varphi_T=X_T\right\}\\
		=&\esssup_{(\bQ,\SQ)\in\CPS(s,T)}\bE_\bQ\left[X_T^1+X_T^2\SQ_T\mid\cF_s\right].
\end{split}
\end{align}
\end{proposition} \noindent

\begin{proof}
	We obtain \eqref{ali:propdualitynonlocal} with the same arguments as in the proof of Proposition \ref{prop:dualitylocal}.
\end{proof}

\begin{theorem}
\label{thm:timeindependent}
Let Assumption \ref{assumption} hold and $s\in[0,T]$. Let $X_T=(X_T^1,X_T^2)$ be a contingent claim such that
\begin{align}
\label{ali:superreplicationdynamiclocaladmtime}
X_T^1+(X_T^2)^+(1-\lambda)S_T-(X_T^2)^-(1+\lambda)S_T\geq -M_s
\end{align}
for some $M_s\in L_+^1(\cF_s,\cQ_\loc)$. Then the following identity holds:
\begin{align*}
\esssup_{(\bQ,\SQ)\in \CPS_{\loc}(s,T)}\bE_\bQ\left[X_T^1+X_T^2\SQ_T\mid\cF_s\right]=\esssup_{(\bQ,\SQ)\in \CPS_{\loc}(0,T)}\bE_\bQ\left[X_T^1+X_T^2\SQ_T\mid\cF_s\right].
\end{align*}
\end{theorem} \noindent

\begin{proof}
	The proof is analogous to the one of Theorem 3.9 of \cite{biagini2019asset}. For further details, see \cite{thesis}.
\end{proof}

\begin{theorem}
\label{thm:timeindependentnonlocal}
Let Assumption \ref{assumption2} hold and $s\in[0,T]$. Let $X_T=(X_T^1,X_T^2)$ be a contingent claim such that 
\begin{align}
\label{ali:superreplicationdynamiclocaladmtimeloc}
X_T^1+(X_T^2)^+(1-\lambda)S_T-(X_T^2)^-(1+\lambda)S_T\geq -M_s^1-M_s^2S_T
\end{align}
for some $(M_s^1,M_s^2)\in L_+^1(\cF_s,\cQ)\times L_+^\infty(\cF_s,\cQ)$. Then the following identity holds:
\begin{align*}
\esssup_{(\bQ,\SQ)\in \CPS(s,T)}\bE_\bQ\left[X_T^1+X_T^2\SQ_T\mid\cF_s\right]=\esssup_{(\bQ,\SQ)\in \CPS(0,T)}\bE_\bQ\left[X_T^1+X_T^2\SQ_T\mid\cF_s\right].
\end{align*}
\end{theorem} \noindent

\begin{proof}
	The proof is analogous to the one of Theorem 3.9 of \cite{biagini2019asset}. For further details, see \cite{thesis}.
\end{proof}\noindent
From now on, we set $\CPS_\loc:=\CPS_\loc(s,T,\lambda)$ and $\CPS:=\CPS(s,T,\lambda)$, respectively. For sake of convenience, we use the following notation for a fixed claim $X_T=(X_T^1,X_T^2)$. We denote 
\begin{align}
	\label{ali:processf}
		F_t:=\esssup_{(\bQ,\SQ)\in\CPS_\loc}\bE_\bQ\left[X_T^1+X_T^2\SQ_T\mid\cF_t\right],\quad t\in[0,T],
\end{align}
and
\begin{align}
	F_t^\bQ:=\bE_\bQ\left[X_T^1+X_T^2\SQ_T\mid\cF_t\right],\quad t\in[0,T].
\end{align}
Next, we provide sufficient conditions such that $F=(F_t)_{t\in [0,T]}$ admits a right-continuous modification.

\subsection{Right continuity}
\label{sec:furtherproperties}

In this section we study the right-continuity of the process $F$ defined in \eqref{ali:processf}. For this purpose we need some preliminary result and further assumptions.

\begin{lemma}
	\label{lem:supremumexpectation}
	Let Assumption \ref{assumption} hold and $s\in[0,T]$. Let $X_T=(X_T^1,X_T^2)\in L^0(\bR^2;\cF_T,\bP)$ such that
\begin{align}
\label{ali:superreplicationdynamiclocalupwards}
X_T^1+(X_T^2)^+(1-\lambda)S_T-(X_T^2)^-(1+\lambda)S_T\geq -M_s
\end{align}
for some $M_s\in L_+^1(\cF_s,\cQ_\loc)$. Then for any $\bQ_0\in\cQ_\loc$ the following identity holds 
\begin{align}
\label{ali:cadlagid}
\bE_{\bQQ}\left[F_t\right]&=\sup_{(\bQ,\SQ)\in\CPS_\loc}\bE_{\bQQ}\left[F_t^\bQ\right],\quad t\in[s,T],
\end{align}
for $F$ defined in \eqref{ali:processf}.
\end{lemma}

\begin{proof}
Let $(\bQQ,\SQQ)\in\CPS_\loc(0,T)$ and $t\in [s,T]$. By monotonicity we immediately obtain
\begin{align}
\bE_{\bQQ}\left[F_t\right]\geq \sup_{(\bQ,\SQ)\in\CPS_\loc}\bE_{\bQQ}\left[F_t^\bQ\right].
\end{align}
For the reverse inequality we use Theorem \ref{thm:timeindependent} to show that 
\begin{align*}
\Phi:=\left\{F_t^\bQ:(\bQ,\SQ)\in\CPS_\loc(t,T)\right\}
\end{align*}
 is directed upwards, see Definition \ref{def:directedupwards}. Let $\bE_\bQ[X_T^1+X_T^2\SQ_T\mid\cF_t],\bE_{\bar{\bQ}}[X_T^1+X_T^2\bSQ_T\mid\cF_t]\in\Phi$. We construct $\bE_{\widehat{\bQ}}[X_T^1+X_T^2\hSQ\mid\cF_t]\in\Phi$ such that $\bE_{\widehat{\bQ}}[X_T^1+X_T^2\hSQ_T\mid\cF_t]\geq \bE_\bQ[X_T^1+X_T^2\SQ_T\mid\cF_t]\vee\bE_{\bar{\bQ}}[X_T^1+X_T^2\bSQ_T\mid\cF_t]$. Define 
\begin{align*}
A_t:=\left\{\bE_\bQ\left[X_T^1+X_T^2\SQ_T\mid\cF_t\right]\geq \bE_{\bar{\bQ}}\left[X_T^1+X_T^2\bSQ_T\mid\cF_t\right]\right\}\in\cF_t.
\end{align*}
Let $Z=(Z^1,Z^2)$ and $\bar{Z}=(\bar{Z}^1,\bar{Z}^2)$ be the processes associated to $(\bQ,\SQ)$ and $(\bbQ,\bSQ)$ respectively, as in Proposition \ref{prop:pmartingale}. Then we define
\begin{align}
\frac{d\widehat{\bQ}}{d\bP}=\frac{\widehat{Z}_T^1}{\bE_\bP\left[\widehat{Z}_T^1\right]}:=\frac{\mathds{1}_{A_t} Z_T^1+\mathds{1}_{A_t^c}\bar{Z}_T^1}{\bE_\bP\left[\mathds{1}_{A_t} Z_T^1+\mathds{1}_{A_t^c}\bar{Z}_T^1\right]},
\end{align}
and for $t\leq u\leq T$,
\begin{align}
\widehat{Z}_u^2:=\mathds{1}_{A_t} Z_u^2+\mathds{1}_{A_t^c}\bar{Z}_u^2
\end{align}
with corresponding
\begin{align}
\hSQ_u=\frac{\widehat{Z}_u^2}{\widehat{Z}_u^1}.
\end{align}
Obviously, $\widehat{Z}=(\widehat{Z}_u^1,\widehat{Z}_u^2)_{t\leq u \leq T}$ satisfies all requirements from Definition \ref{def:cpspolar}, i.e., $\widehat{Z}\in\mathcal{Z}_\loc(t,T)$. Clearly, $(1-\lambda)S_u\leq \hSQ_u\leq (1+\lambda)S_u$ for all $u\in[t,T]$. For the local martingale property let $(\tau_n)_{n\in\bN}$ be a localizing sequence for $\SQ$ and $\bSQ$. For $t\leq u\leq v\leq T$ we get
\begin{align*}
&\bE_{\widehat{\bQ}}\left[\left(\hSQ_v\right)^{\tau_n}\mid\cF_u\right]=\bE_\bP\left[\left(\frac{\widehat{Z}_v^2}{\widehat{Z}_v^1}\right)^{\tau_n}\frac{\widehat{Z}_T^1}{\bE_\bP\left[\widehat{Z}_T^1\right]}\mid\cF_u\right]\frac{\bE_\bP\left[\widehat{Z}_T^1\right]}{\widehat{Z}_{u\wedge\tau_n}^1}\\
=&\bE_\bP\left[\left(\mathds{1}_{A_t}Z_v^1+\mathds{1}_{A_t^c}\bar{Z}_v^2\right)^{\tau_n}\mid\cF_u\right]\frac{1}{\widehat{Z}_{u\wedge\tau_n}^1}=\left(\mathds{1}_{A_t}\bE_\bP\left[(Z_v^2)^{\tau_n}\mid\cF_u\right]+\mathds{1}_{A_t^c}\bE_\bP\left[(\bar{Z}_v^2)^{\tau_n}\mid\cF_u\right]\right)\frac{1}{\widehat{Z}_{u\wedge\tau_n}^1}\\
=&\left(\mathds{1}_{A_t}Z_{u\wedge\tau_n}^2+\mathds{1}_{A_t^c}\bar{Z}_{u\wedge\tau_n}^2\right)\frac{1}{\widehat{Z}_{u\wedge\tau_n}^1}=\left(\widehat{S}_u\right)^{\tau_n},
\end{align*}
where we used that $\mathds{1}_{A_t},\mathds{1}_{A_t^c}$ are measurable for $\cF_t\subset \cF_u$. In particular, by Theorem A.33 of \cite{follmerschied}, there exists an increasing sequence $(\bE_{\bQ^n}[X_T^1+X_T^2\tSn_T\mid\cF_t])_{n\in\bN}\subset\Phi$ such that
\begin{align}
F_t=\esssup_{(\bQ,\SQ)\in\CPS_\loc(t,T)}\bE_\bQ\left[X_T^1+X_T^2\SQ_T\mid\cF_t\right]=\lim_{n\to \infty}\bE_{\bQ^n}\left[X_T^1+X_T^2\tSn_T\mid\cF_t\right].
\end{align}
By the Theorem of Monotone Convergence we obtain
\begin{align}
\nonumber
\bE_{\bQQ}\left[F_t\right]&=\lim_{n\to\infty}\bE_{\bQQ}\left[\bE_{\bQ^n}\left[X_T^1+X_T^2\tSn_T\mid\cF_t\right]\right]\\
\nonumber
&\leq \sup_{(\bQ,\SQ)\in\CPS_\loc(t,T,\lambda)}\bE_{\bQQ}\left[\bE_\bQ\left[X_T^1+X_T^2\SQ_T\mid\cF_t\right]\right]\\
\label{eq:kramkov1}
&=\sup_{(\bQ,\SQ)\in\CPS_\loc(0,T,\lambda)}\bE_{\bQQ}\left[\bE_\bQ\left[X_T^1+X_T^2\SQ_T\mid\cF_t\right]\right].
\end{align}
The last equality in \eqref{eq:kramkov1} holds due to similar arguments as in the proof of Theorem \ref{thm:timeindependent}. This concludes the proof of \eqref{ali:cadlagid}.
\end{proof}\noindent

\begin{lemma}
	\label{lem:supremumexpectationnonlocal}
	Let Assumption \ref{assumption2} hold and $s\in[0,T]$. Let $X_T=(X_T^1,X_T^2)\in L^0(\bR^2;\cF_T,\bP)$ such that
\begin{align}
\label{ali:superreplicationdynamicnonlocalupwards}
X_T^1+(X_T^2)^+(1-\lambda)S_T-(X_T^2)^-(1+\lambda)S_T\geq -M_s^1-M_s^2S_T
\end{align}
for some $(M_s^1,M_s^2)\in L_+^1(\cF_s,\cQ)\times L_+^\infty(\cF_s,\cQ)$. Then for any $\bQ_0\in\cQ$ the following identity holds 
\begin{align}
\bE_{\bQQ}\left[\esssup_{(\bQ,\SQ)\in\CPS}\bE_\bQ\left[X_T^1+X_T^2\SQ_T\mid\cF_t\right]\right]=\sup_{(\bQ,\SQ)\in\CPS}\bE_{\bQQ}\left[F_t^\bQ\right],\quad t\in[s,T].
\end{align}
\end{lemma}

\begin{proof}
	The arguments are identical to the proof of Lemma \ref{lem:supremumexpectation}.
\end{proof}\noindent
Let $(\sigma_n)_{n\in\bN}$ be a sequence of decreasing, $[0,T]$-valued stopping times with $\sigma_n\downarrow \sigma=\sigma_\infty$ as $n$ tends to infinity. In the sequel we set $\bar\bN:=\bN\cup \{\infty\}$.\\
Lemma \ref{lem:supremumexpectation} shows that for any $n\in\bar{\bN}$ there exists a sequence $(\bQ(m_k(n)),\SQ(m_k(n)))_{k\in\bN}\subset\CPS_\loc(0,T)$ such that $\bE_\bQQ[F_{\sigma_n}^{\bQ(m_k(n))}]\uparrow \bE_\bQQ[F_{\sigma_n}]$ as $k$ tends to infinity. Further, it is easy to see that these sequences can be taken uniformly over $n\in\bar\bN$. For $n\in\bar\bN$ take the subsequence $(m_{k_l}(n))_{l\in\bN}\subset (m_k(n))_{k\in\bN}$ defined by
\begin{align*}
	m_{k_1}(n)&:=\inf\left\{k\geq 1: \left|\bE_\bQQ\left[F_{\sigma_n}^{\bQ(m_k(n))}\right]-\bE_\bQQ\left[F_{\sigma_n}\right]\right|<1\right\},\\
	m_{k_l}(n)&:=\inf\left\{k > m_{k_{l-1}}(n): \left|\bE_\bQQ\left[F_{\sigma_n}^{\bQ(m_k(n))}\right]-\bE_\bQQ\left[F_{\sigma_n}\right]\right|<\frac{1}{l}\right\}.
\end{align*}

\begin{assumption}
	\label{assumptioncadlag}
	We assume the existence of $\bQ_0\in\cQ_\loc$ such that for any decreasing sequence of stopping times $0\leq(\sigma_n)_{n\in\bN}\leq T$ with $\sigma_n\downarrow \sigma$ as $n$ tends to infinity, there exists a sequence $$(\bQ(m_k(n)),\SQ(m_k(n))_{k\in\bN}\subset\CPS_\loc(0,T),\ n\in\bar\bN,$$ such that $(\bE_\bQQ[F_{\sigma_n}^{\bQ(m_k(n))}])_{k\in\bN}$ converges uniformly over all $n\in\bar\bN$ to $\bE_\bQQ[F_{\sigma_n}]$, i.e., for all $\epsilon>0$ there exists $K=K(\epsilon)\in\bN$ such that
\begin{align}
\label{ali:assumption3sequence}
	\left|\bE_\bQQ\left[F_{\sigma_n}^{\bQ(m_k(n))}\right]-\bE_\bQQ\left[F_{\sigma_n}\right]\right|<\epsilon,\text{ for all }k\geq K,\text{ and for all }n\in\bar\bN,
\end{align}	
	and that for all $k\in\bN$, $\bQQ(\bigcup_{N\in\bN}A_N^{\epsilon,k})=1 $, where
\begin{align}
\label{ali:assumption3setan}
	A_N^{\epsilon,k}:=\left\{\omega\in\Omega: |F_{\sigma_n}^{\bQ(m_k(n_0))}-F_{\sigma}^{\bQ(m_k(n_0))}|(\omega)<\epsilon,\ \forall n\geq N,\ \forall n_0\in \bar\bN\right\}.
\end{align}
\end{assumption}\noindent
Assumption \ref{assumptioncadlag} can be thought as equi-continuity in time at level $k$, of a family of approximating sequences of consistent price systems. A uniformly approximating sequence does always exist, see Lemma \ref{lem:supremumexpectation} and the comment thereafter, but in general it is not true that $\bQQ(\bigcup_{N\in\bN}A_N^{\epsilon,k})=1 $ for $A_N^{\epsilon,k}$ given in \eqref{ali:assumption3setan}. This is the key feature of Assumption \ref{assumptioncadlag}.
\begin{remark}
	In the special case of $X_T=(0,1)$ Assumption \ref{assumptioncadlag} reads as follows: Fix an approximating sequence such that \eqref{ali:assumption3sequence} holds. Then 
	\begin{equation}
\label{ali:assumption3setan}
	A_N^{\epsilon,k}:=\Set{\omega\in\Omega:\begin{array}{l}\left|\bE_{\bQ(m_k(n_0))}\left[\tilde{S}_T^{\bQ(m_k(n_0))}\mid\cF_{\sigma_n}\right]-\bE_{\bQ(m_k(n_0))}\left[\tilde{S}_T^{\bQ(m_k(n_0))}\mid\cF_\sigma\right]\right|(\omega)<\epsilon,\\ \forall n\geq N,\ \forall n_0\in \bar\bN\end{array}}
\end{equation}
\end{remark}

\begin{theorem}
\label{thm:cadlag}
Suppose that Assumption \ref{assumption} and \ref{assumptioncadlag} hold and $s\in[0,T]$. Let $X_T=(X_T^1,X_T^2)\in L^\infty(\cF_T,\bP)\times L^\infty(\cF_T,\bP)$ such that
\begin{align}
\label{ali:superreplicationdynamiclocaladmcadlag}
X_T^1+(X_T^2)^+(1-\lambda)S_T-(X_T^2)^-(1+\lambda)S_T\geq -M_s
\end{align}
for some $M_s\in L_+^1(\cF_s,\cQ_\loc)$. Then $F$ in \eqref{ali:processf} admits a right-continuous modification with respect to $\bP$.
\end{theorem}

\begin{proof}
Let $\bQ_0\in\cQ_\loc$ be the measure given by Assumption \ref{assumptioncadlag}. We now prove that $F$ in \eqref{ali:processf} admits a right-continuous modification with respect to $\bQ_0$. Since all measure $\bQ\in\cQ_\loc$ are equivalent to $\bP$, this is equivalent to show that $F$ admits a right-continuous modification with respect to $\bP$.\\
By Theorem 48 in \cite{dellacherieb}, the paths of $F$ are right-continuous (outside an evanescent set), if $\lim_{n\to\infty}\bE_{\bP}\left[F_{\sigma_n}\right]=\bE_{\bP}\left[F_{\lim_{n\to\infty}\sigma_n}\right]$ for every decreasing sequence $(\sigma_n)_{n\in\bN}$ of bounded stopping times.\\
Let $(\sigma_n)_{n\in\bN}$ be a decreasing sequence of stopping times with values in $[0,T]$ such that $\sigma_n\downarrow \sigma$ as $n$ tends to infinity. We now prove that
\begin{align*}
\lim_{n\to\infty} \bE_{\bQQ}\left[F_{\sigma_n}\right]=\bE_{\bQQ}\left[F_\sigma\right].
\end{align*}
Because $(X_T^1,X_T^2)\in L^\infty(\cF_T,\bP)\times L^\infty(\cF_T,\bP)$, there exists $C_1,C_2\in\bR$ such that $|X_T^1|\leq C_1$ and $|X_T^2|\leq C_2$. For $t\in[0,T]$ we thus have,
\begin{align}
\nonumber
	|F_t|&\leq \esssup_{(\bQ,\SQ)\in\CPS_\loc}\bE_\bQ\left[\left|X_T^1+X_T^2\SQ_T\right|\mid\cF_t\right]\\
	\nonumber
	&\leq  \esssup_{(\bQ,\SQ)\in\CPS_\loc}\bE_\bQ[C_1+C_2\SQ_T\mid\cF_t]\\
	\nonumber
	&=C_1+C_2 \esssup_{(\bQ,\SQ)\in\CPS_\loc}\bE_\bQ[\SQ_T\mid\cF_t]\\
	\nonumber
	&\leq C_1+C_2 (1+\lambda)S_t	\\
	\label{ali:linftybound}
	&\leq C_1+C_2 \frac{1+\lambda}{1-\lambda} \tilde S ^\bQQ_t.
\end{align}
We prove that the family
\begin{align*}
	\cG:=\left\{|F_{\sigma_n}^\bQ-F_\sigma^\bQ|:n\in\bN,\ (\bQ,\SQ)\in\CPS_\loc(0,T)\right\}
\end{align*}
is uniformly integrable with respect to $\bQQ$. First note, that
\begin{align}
	|F_{\sigma_n}^\bQ-F_\sigma^\bQ|\leq |F_{\sigma_n}^\bQ|+|F_\sigma^\bQ|\leq |F_{\sigma_n}|+|F_\sigma|.
\end{align}
It is easy to see that $F_\sigma\in L^1(\cF_\sigma,\bQQ)$ because of \eqref{ali:linftybound}. Equation \eqref{ali:linftybound} also implies 
\begin{align*}
	|F_{\sigma_n}|\leq C_1+C_2\frac{1+\lambda}{1-\lambda}\tilde S_{\sigma_n}^\bQQ,\quad\text{for all }n\in\bN.
\end{align*}
Therefore, it is enough to prove $\{\tilde S_{\sigma_n}^\bQQ:n\in\bN\}$ is uniformly integrable with respect to $\bQQ$. To this purpose, we first show that $\SQQ_{\sigma_n}\xrightarrow{L^1} \SQQ_\sigma$. Because $\SQQ$ is a non-negative, c\`adl\`ag local $\bQQ$-martingale, $\SQQ$ is also a supermartingale under $\bQQ$ and we get by Theorem 9 of \cite{protterstochastic} that
\begin{align}
\label{ali:l1conv}
	\lim_{n\to\infty}\bE_\bQQ[\SQQ_{\sigma_n}]=\bE_\bQQ[\SQQ_\sigma].
\end{align}
Thus, Scheff\`e's Lemma guarantees that $\SQQ_{\sigma_n}\xrightarrow{L^1} \SQQ_\sigma$. Since $\SQQ_{\sigma_n}\in L^1(\cF_T,\bQQ)$ for all $n\in \bN$ and $\SQQ_{\sigma_n}\xrightarrow{a.s.}\SQQ_\sigma$ and $\SQQ_{\sigma_n}\xrightarrow{L^1}\SQQ_\sigma$, Theorem 6.25 of \cite{klenke2013probability} implies that $\{\SQQ_{\sigma_n}:n\in\bN\}$ is uniformly integrable with respect to $\bQQ$ which yields that $\cG$ is uniformly integrable with respect to $\bQQ$.\\
By Assumption \ref{assumptioncadlag} there exists for all $\epsilon>0$ and for each $n\in\bar\bN$ a sequence 
$$(\bQ(m_k(n)),\SQ(m_k(n)))_{k\in\bN}\subset\CPS_\loc(0,T)$$ such that \eqref{ali:assumption3sequence} is satisfied. Fix $\epsilon>0$. By $(\bE_\bQQ[F_\sigma^{\bQ(m_k(\infty)}])_{k\in\bN}$ we denote the sequence converging to $\bE_\bQQ[F_\sigma]$. For $N\in\bN$ and $k\in\bN$ consider the set $A_N^k=A_N^{\epsilon/8,k}$ defined in \eqref{ali:assumption3setan}, i.e.,
\begin{align}
\label{ali:setan}
	A_N^k=\left\{\omega\in\Omega: |F_{\sigma_n}^{\bQ(m_k(n_0))}-F_{\sigma}^{\bQ(m_k(n_0))}|(\omega)<\frac{\epsilon}{8},\ \forall n\geq N,\ \forall n_0\in \bar\bN\right\}.
\end{align}
By Assumption \ref{assumptioncadlag} we get that $\bQQ(\bigcup_{N\in\bN}A_N^k)=1$ for all $k\in\bN$. As for fixed $k\in\bN$ we get $A_N^k\subset A_{N+1}^k$ we can conclude that $\bQQ(A_N^k)\uparrow 1$ as $N$ tends to infinity.\\
Fix $k\in\bN$ such that 
\begin{align}
\label{ali:uniformsequence}
	\left|\bE_\bQQ\left[F_{\sigma_n}^{\bQ(m_k(n))}\right]-\bE_\bQQ\left[F_{\sigma_n}\right]\right|<\frac{\epsilon}{8},\text{ for all }n\in\bar\bN.
\end{align}	
Since $\cG$ is uniformly integrable with respect to $\bQQ$, there exists $\delta=\delta(\epsilon)$ such that for all $\Lambda\in\cF_T$ satisfying $\bQQ(\Lambda)<\delta$, we get 
\begin{align}
\label{ali:uniformintegrabilitygeneral}
	\bE_\bQQ\left[\left|F_{\sigma_n}^{\bQ(m_k(n_0))}-F_\sigma^{\bQ(m_k(n_0))}\right|\mathds{1}_\Lambda \right]<\frac{\epsilon}{8},
\end{align}
for all $n,n_0\in\bar\bN$. Since $\bQQ(A_N^k) \uparrow 1$ as $N$ tends to infinity, there exists $N_0=N_0(\epsilon,k)\in\bN$ such that $\bQQ((A_N^k)^c)<\delta$ for all $N\geq N_0$. Fix $N\geq N_0$ and let $n\geq N$. Then we have
\begin{align}
\nonumber
	\bE_\bQQ\left[\left|F_{\sigma_n}^{\bQ(m_k(n_0))}-F_\sigma^{\bQ(m_k(n_0))}\right|\right]&=\bE_\bQQ\left[\left|F_{\sigma_n}^{\bQ(m_k(n_0))}-F_\sigma^{\bQ(m_k(n_0))}\right|\mathds{1}_{A_N^k}\right]\\
	\nonumber
	&+\bE_\bQQ\left[\left|F_{\sigma_n}^{\bQ(m_k(n_0))}-F_\sigma^{\bQ(m_k(n_0))}\right|\mathds{1}_{(A_N^k)^c} \right]\\
	\label{ali:indicatordecomposition}
	&<\frac{\epsilon}{4},
\end{align}
by the definition of the set $A_N^k$ in \eqref{ali:setan} and by \eqref{ali:uniformintegrabilitygeneral} because $\bQQ((A_N^k)^c)<\delta$. We consider three different cases.

	\begin{enumerate}
		\item [Case 1:] Assume 
		\begin{align}
		\label{case1}
			\bE_\bQQ\left[F_{\sigma_n}^{\bQ(m_k(n))}\right]\leq \bE_\bQQ\left[F_{\sigma_n}^{\bQ(m_k(\infty))}\right].
		\end{align}
	Then we have
		\begin{align}
		\nonumber
		& \left|\bE_\bQQ\left[F_{\sigma_n}\right]-\bE_\bQQ\left[F_{\sigma}\right]\right| \\	
		\nonumber
		\leq & \left|\bE_\bQQ\left[F_{\sigma_n}\right]- \bE_\bQQ\left[F_{\sigma_n}^{\bQ(m_k(\infty))}\right]\right|+\left|\bE_\bQQ\left[F_{\sigma_n}^{\bQ(m_k(\infty))}\right]-\bE_\bQQ\left[F_{\sigma}^{\bQ(m_k(\infty))}\right]\right|\\
		\nonumber
		+& \left|\bE_\bQQ\left[F_\sigma^{\bQ(m_k(\infty))}\right]-\bE_\bQQ\left[F_{\sigma}\right]\right|\\
		\nonumber
		< & \left|\bE_\bQQ\left[F_{\sigma_n}\right]- \bE_\bQQ\left[F_{\sigma_n}^{\bQ(m_k(n))}\right]\right|+\left|\bE_\bQQ\left[F_{\sigma_n}^{\bQ(m_k(\infty))}\right]-\bE_\bQQ\left[F_{\sigma}^{\bQ(m_k(\infty))}\right]\right|+\frac{\epsilon}{8}\\
		\label{ali:case1main}
		<& \frac{2\epsilon}{8}+\bE_\bQQ\left[\left|F_{\sigma_n}^{\bQ(m_k(\infty))}-F_{\sigma}^{\bQ(m_k(\infty))}\right|\mathds{1}_{A_N^k}\right]+\bE_\bQQ\left[\left|F_{\sigma_n}^{\bQ(m_k(\infty))}-F_{\sigma}^{\bQ(m_k(\infty))}\right|\mathds{1}_{(A_N^k)^c}\right]\\
		\nonumber
		< & \frac{4\epsilon}{8}<\epsilon
		\end{align}
		The second inequality holds due to \eqref{case1} and to the fact that $F_{\sigma_n}$ is the essential supremum over all consistent price systems, and because of \eqref{ali:uniformsequence}. Also \eqref{ali:case1main} holds due to  \eqref{ali:uniformsequence}. In the last step we applied \eqref{ali:indicatordecomposition}. 
		
		\item [Case 2:] Assume 
		\begin{align} 
		\label{case2}
		\bE_\bQQ\left[F_{\sigma}^{\bQ(m_k(n))}\right]\geq \bE_\bQQ\left[F_{\sigma}^{\bQ(m_k(\infty))}\right].
		\end{align}
	Then we have
		\begin{align*}
		& \left|\bE_\bQQ\left[F_{\sigma_n}\right]-\bE_\bQQ\left[F_{\sigma}\right]\right| \\	
		\leq & \left|\bE_\bQQ\left[F_{\sigma_n}\right]- \bE_\bQQ\left[F_{\sigma_n}^{\bQ(m_k(n))}\right]\right|+\left|\bE_\bQQ\left[F_{\sigma_n}^{\bQ(m_k(n))}\right]-\bE_\bQQ\left[F_{\sigma}^{\bQ(m_k(n))}\right]\right|\\
		+ & \left|\bE_\bQQ\left[F_\sigma^{\bQ(m_k(n))}\right]-\bE_\bQQ\left[F_{\sigma}\right]\right|\\
		<& \frac{\epsilon}{8}+\left|\bE_\bQQ\left[F_{\sigma_n}^{\bQ(m_k(n))}\right]-\bE_\bQQ\left[F_{\sigma}^{\bQ(m_k(n))}\right]\right|+\left|\bE_\bQQ\left[F_\sigma^{\bQ(m_k(\infty))}\right]-\bE_\bQQ\left[F_{\sigma}\right]\right|\\
		< & \frac{2\epsilon}{8}+\bE_\bQQ\left[\left|F_{\sigma_n}^{\bQ(m_k(n))}-F_{\sigma}^{\bQ(m_k(n))}\right|\mathds{1}_{A_N^k}\right]+\bE_\bQQ\left[\left|F_{\sigma_n}^{\bQ(m_k(n))}-F_{\sigma}^{\bQ(m_k(n))}\right|\mathds{1}_{(A_N^k)^c}\right]\\
		< & \frac{4\epsilon}{8}<\epsilon
		\end{align*}
		The steps in Case 2 are analogously to Case 1, replacing \eqref{case1} by \eqref{case2}.
		
		\item [Case 3:] Assume
		\begin{align}
		\label{case3}
		\bE_\bQQ\left[F_{\sigma_n}^{\bQ(m_k(n))}\right]> \bE_\bQQ\left[F_{\sigma_n}^{\bQ(m_k(\infty))}\right]\text{ and }
		\bE_\bQQ\left[F_{\sigma}^{\bQ(m_k(n))}\right]< \bE_\bQQ\left[F_{\sigma}^{\bQ(m_k(\infty))}\right].
	\end{align}
Then we have
		\begin{align*}
		& \left|\bE_\bQQ\left[F_{\sigma_n}\right]-\bE_\bQQ\left[F_{\sigma}\right]\right| \\	
		\leq & \left|\bE_\bQQ\left[F_{\sigma_n}\right]- \bE_\bQQ\left[F_{\sigma_n}^{\bQ(m_k(n))}\right]\right|+\left|\bE_\bQQ\left[F_{\sigma_n}^{\bQ(m_k(n))}\right]-\bE_\bQQ\left[F_{\sigma}^{\bQ(m_k(\infty))}\right]\right|\\
		+& \left|\bE_\bQQ\left[F_\sigma^{\bQ(m_k(\infty))}\right]-\bE_\bQQ\left[F_{\sigma}\right]\right|\\
		<& \frac{\epsilon}{8	}+\left|\bE_\bQQ\left[F_{\sigma_n}^{\bQ(m_k(n))}\right]-\bE_\bQQ\left[F_{\sigma}^{\bQ(m_k(\infty))}\right]\right|+\frac{\epsilon}{8}
		\end{align*}
		The second inequality holds due to \eqref{ali:uniformsequence}. For the remaining part we obtain
		\begin{align*}
			&\left|\bE_\bQQ\left[F_{\sigma_n}^{\bQ(m_k(n))}\right]-\bE_\bQQ\left[F_{\sigma}^{\bQ(m_k(\infty))}\right]\right|\\
			\leq & \left|\bE_\bQQ\left[F_{\sigma_n}^{\bQ(m_k(n))}\right]-\bE_\bQQ\left[F_\sigma^{\bQ(m_k(n))}\right]\right|+\left|\bE_\bQQ\left[F_\sigma^{\bQ(m_k(n))}\right]-\bE_\bQQ\left[F_\sigma^{\bQ(m_k(\infty)}\right]\right|\\
			< & \frac{2\epsilon}{8}+\left|\bE_\bQQ\left[F_\sigma^{\bQ(m_k(\infty))}\right]-\bE_\bQQ\left[F_\sigma^{\bQ(m_k(n))}\right]\right|
		\end{align*}
		by the triangle inequality and \eqref{ali:indicatordecomposition}. Then \eqref{case3}, and \eqref{ali:indicatordecomposition} imply
		\begin{align*}
			&\left|\bE_\bQQ\left[F_\sigma^{\bQ(m_k(\infty))}\right]-\bE_\bQQ\left[F_\sigma^{\bQ(m_k(n))}\right]\right|=\bE_\bQQ\left[F_\sigma^{\bQ(m_k(\infty))}\right]-\bE_\bQQ\left[F_\sigma^{\bQ(m_k(n))}\right]\\
			<& \bE_\bQQ\left[F_\sigma^{\bQ(m_k(\infty))}\right]-\bE_\bQQ\left[F_{\sigma_n}^{\bQ(m_k(n))}\right]+\frac{2\epsilon}{8}\\
			<&\bE_\bQ\left[F_\sigma^{\bQ(m_k(\infty))}\right]-\bE_\bQQ\left[F_{\sigma_n}^{\bQ(m_k(n))}\right]+\frac{2\epsilon}{8}<\frac{4\epsilon}{8}.
		\end{align*}
		Combining these estimations we obtain
		
		\begin{align*}
			 \left|\bE_\bQQ\left[F_{\sigma_n}\right]-\bE_\bQQ\left[F_{\sigma}\right]\right| <  \frac{4\epsilon}{8}+\left|\bE_\bQQ\left[F_\sigma^{\bQ(m_k(\infty))}\right]-\bE_\bQQ\left[F_\sigma^{\bQ(m_k(n))}\right]\right|< \epsilon.
		\end{align*}
	\end{enumerate}
	Summarizing, we have	
	\begin{align*}
		\left|\bE_\bQQ\left[F_{\sigma_n}\right]-\bE_\bQQ\left[F_{\sigma}\right]\right| <\epsilon.
	\end{align*}
	Since $\epsilon>0$ was arbitrary, we can conclude that
	\begin{align*}
		\bE_\bQQ\left[F_{\sigma_n}\right]\xrightarrow{n\to\infty} \bE_\bQQ\left[F_{\sigma}\right].
	\end{align*}	
	This implies that $\left(F_t\right)_{t\in[0,T]}$ admits a right-continuous modification with respect to $\bQQ$ and hence also with respect to $\bP$.
\end{proof}\noindent
It is easy to see that Assumption \ref{assumptioncadlag} and Theorem \ref{thm:cadlag} can be analogously formulated for the non-local case. 

\begin{assumption}
	\label{assumptioncadlagnonlocal}
	We assume the existence of $\bQ_0\in\cQ$ such that for any decreasing sequence of stopping times $0\leq(\sigma_n)_{n\in\bN}\leq T$ with $\sigma_n\downarrow \sigma$ as $n$ tends to infinity, there exists a sequence $$(\bQ(m_k(n)),\SQ(m_k(n))_{k\in\bN}\subset\CPS(0,T),\ n\in\bar\bN,$$ such that $(\bE_\bQQ[F_{\sigma_n}^{\bQ(m_k(n))}])_{k\in\bN}$ converges uniformly over all $n\in\bar\bN$ to $\bE_\bQQ[F_{\sigma_n}]$, i.e., for all $\epsilon>0$ there exists $K=K(\epsilon)\in\bN$ such that
\begin{align}
\label{ali:assumption3sequencenonlocal}
	\left|\bE_\bQQ\left[F_{\sigma_n}^{\bQ(m_k(n))}\right]-\bE_\bQQ\left[F_{\sigma_n}\right]\right|<\epsilon,\text{ for all }k\geq K,\text{ and for all }n\in\bar\bN,
\end{align}	
	and that for all $k\in\bN$, $\bQQ(\bigcup_{N\in\bN}A_N^{\epsilon,k})=1 $, where
\begin{align}
\label{ali:assumption3setannonlocal}
	A_N^{\epsilon,k}:=\left\{\omega\in\Omega: |F_{\sigma_n}^{\bQ(m_k(n_0))}-F_{\sigma}^{\bQ(m_k(n_0))}|(\omega)<\epsilon,\ \forall n\geq N,\ \forall n_0\in \bar\bN\right\}.
\end{align}
\end{assumption}\noindent

\begin{theorem}
\label{thm:cadlagnonlocal}
Let Assumption \ref{assumption2} and \ref{assumptioncadlagnonlocal} hold and $s\in[0,T]$. Let $X_T=(X_T^1,X_T^2)\in L^\infty(\bR^2;\cF_T,\bP)$ such that
\begin{align}
\label{ali:superreplicationdynamiclocaladmcadlagnonlocal}
X_T^1+(X_T^2)^+(1-\lambda)S_T-(X_T^2)^-(1+\lambda)S_T\geq -M_s^1-M_s^2S_T
\end{align}
for some $(M_s^1,M_s^2)\in L_+^1(\cF_s,\cQ)\times L_+^\infty(\cF_s,\cQ)$. Then $F$ defined in \eqref{ali:processf} admits a right-continuous modification with respect to $\bP$.
\end{theorem}

\begin{proof}
	The proof of Theorem \ref{thm:cadlag} carries over using Assumptions \ref{assumption2} and \ref{assumptioncadlagnonlocal}.
\end{proof}\noindent
Theorem \ref{thm:cadlag} (resp. Theorem \ref{thm:cadlagnonlocal}) guarantees that the process $F$, defined in \eqref{ali:processf}, is well-defined. It is important to note the difference between $F$ and the super-replication price in the frictionless setting, see \cite{quenez}, \cite{kramkovduality}. In a frictionless market model the wealth of a portfolio equals the liquidation value and the price to buy the portfolio.\\
Under the presence of transaction costs this does no longer hold. More precisely, the liquidation value of a portfolio and the buying price of a portfolio are in general not the same. The process $F$ defines the capital that is needed to start a self-financing, admissible strategy in order to super-replicate the contingent claim. This price is usually higher than the liquidation value.

\begin{example}
\label{ex:assumption23}
	Suppose Assumption \ref{assumption2} holds and let $(X_T^1,X_T^2)=(0,1)$. Then, for each $k\in\bN$ such that $\frac{1}{k}\leq \lambda$ there exists $(\bQ(k),\SQ(k))\in\CPS(0,T,\frac{1}{k})$. We get
	\begin{align}
		(1-\lambda)S_t\leq \SQ_t(k) \leq \frac{1+\lambda}{1+\frac{1}{k}}\SQ_t(k)\leq (1+\lambda)S_t.
	\end{align}
	In particular, $(\bQ(k),\mu_k\SQ(k))\in\CPS(0,T,\lambda)$ for $\mu_k:=\frac{1+\lambda}{1+\frac{1}{k}}$. Furthermore, we have by Proposition 3.11 of \cite{biagini2019asset} that
	\begin{align}
		F_t=(1+\lambda)S_t,\quad t\in[0,T].
	\end{align}
	By the martingale property of $\SQ(k)$ we obtain
	\begin{align}
	\nonumber
		&\left| (1+\lambda)S_t - \bE_{\bQ(k)}\left[\mu_k\SQ_T(k)\mid\cF_t\right]\right|=\left|(1+\lambda)S_t - \mu_k\SQ_t(k)\right|\\
		\leq &\left|(1+\lambda)S_t-\mu_k \left(1-\frac{1}{k}\right)S_t\right|= \left|(1+\lambda)S_t\left( 1- \frac{1-\frac{1}{k}}{1+\frac{1}{k}}\right)\right|=(1+\lambda)S_t\frac{2}{k+1}.
	\end{align}
	Then we get for any $\bQQ\in\cQ$, and $t\in[0,T]$ that
	\begin{align}
		\nonumber
		&\left|\bE_\bQQ[F_t]-\bE_\bQQ\left[\mu_kF_t^{\bQ(k)}\right]\right|\leq \bE_\bQQ\left[\left| (1+\lambda)S_t - \bE_{\bQ(k)}\left[\mu_k\SQ_T(k)\mid\cF_t\right]\right|\right]\\
		\nonumber
		\leq &\bE_\bQQ\left[(1+\lambda)S_t\frac{2}{k+1}\right]\leq \bE_\bQQ\left[\frac{1+\lambda}{1-\lambda}\frac{2}{k+1}\SQQ_t\right]= \frac{1+\lambda}{1-\lambda}\frac{2}{k+1}\SQQ_0\\
		\label{ali:exampleuniformsequence}
		\leq &\frac{(1+\lambda)^2}{1-\lambda}\frac{2}{k+1}S_0.
	\end{align}
	Therefore, we can easily see that for for every $\epsilon>0$ there exists $k\in\bN$ such that \eqref{ali:assumption3sequencenonlocal} of Assumption \ref{assumptioncadlagnonlocal} is fulfilled by the sequence $(\bQ(k),\mu_k\SQ(k))_{k\in \bN}\subset\CPS(0,T)$ which is independent of $t\in[0,T]$. Let $(\sigma_n)_{n\in\bN}$ be any decreasing sequence of stopping times. Note that in this case we get that $A_N^k$ does not depend on $n_0\in\bN$ anymore, i.e.,
	\begin{align}
		A_N^k =\left\{\omega\in\Omega: |F_{\sigma_n}^{\bQ(k)}-F_{\sigma}^{\bQ(k)}|(\omega)<\epsilon,\ \forall n\geq N\right\}.
	\end{align}
	By Definition \ref{def:cps}, all consistent price systems are c\`adl\`ag which yields to $\bQQ(\bigcup_{N\in\bN}A_N^k)=1$. We conclude that under Assumption \ref{assumption2}, Assumption \ref{assumptioncadlagnonlocal} is fulfilled for $X_T=(0,1)$. 
\end{example}

\bibliography{superreplication}

\begin{thebibliography}{28}
\providecommand{\natexlab}[1]{#1}
\providecommand{\url}[1]{\texttt{#1}}
\expandafter\ifx\csname urlstyle\endcsname\relax
  \providecommand{\doi}[1]{doi: #1}\else
  \providecommand{\doi}{doi: \begingroup \urlstyle{rm}\Url}\fi

\bibitem[Bartl et~al.(2019)Bartl, Kupper, Pr{\"o}mel, and
  Tangpi]{bartl2019duality}
Daniel Bartl, Michael Kupper, David~J Pr{\"o}mel, and Ludovic Tangpi.
\newblock Duality for pathwise superhedging in continuous time.
\newblock \emph{Finance and Stochastics}, 23\penalty0 (3):\penalty0 697--728,
  2019.

\bibitem[Bartl et~al.(2020)Bartl, Kupper, and Neufeld]{bartl2020pathwise}
Daniel Bartl, Michael Kupper, and Ariel Neufeld.
\newblock Pathwise superhedging on prediction sets.
\newblock \emph{Finance and Stochastics}, 24\penalty0 (1):\penalty0 215--248,
  2020.

\bibitem[Biagini and Reitsam(2019)]{biagini2019asset}
Francesca Biagini and Thomas Reitsam.
\newblock Asset price bubbles in market models with proportional transaction
  costs.
\newblock \emph{Preprint LMU}, 2019.

\bibitem[Burzoni et~al.(2017)Burzoni, Frittelli, Maggis,
  et~al.]{burzoni2017model}
Matteo Burzoni, Marco Frittelli, Marco Maggis, et~al.
\newblock Model-free superhedging duality.
\newblock \emph{Annals of Applied Probability}, 27\penalty0 (3):\penalty0
  1452--1477, 2017.

\bibitem[Campi and Schachermayer(2006)]{campischachermayer}
Luciano Campi and Walter Schachermayer.
\newblock A super-replication theorem in {K}abanov's model of transaction
  costs.
\newblock \emph{Finance and Stochastics}, 10\penalty0 (4):\penalty0 579--596,
  2006.

\bibitem[Carassus et~al.(2007)Carassus, Gobet, and Temam]{carassus2007class}
Laurence Carassus, Emmanuel Gobet, and Emmanuel Temam.
\newblock A class of financial products and models where super-replication
  prices are explicit.
\newblock In \emph{Stochastic Processes and Applications to Mathematical
  Finance}, pages 67--84. World Scientific, 2007.

\bibitem[Cvitani{\'c} and Karatzas(1996)]{cvitanic}
Jak{\v{s}}a Cvitani{\'c} and Ioannis Karatzas.
\newblock Hedging and portfolio optimization under transaction costs: A
  martingale approach.
\newblock \emph{Mathematical Finance}, 6\penalty0 (2):\penalty0 133--165, 1996.

\bibitem[Cvitani{\'c} et~al.(1999)Cvitani{\'c}, Pham, and
  Touzi]{cvitanic1999closed}
Jak{\v{s}}a Cvitani{\'c}, Huyen Pham, and Nizar Touzi.
\newblock A closed-form solution to the problem of super-replication under
  transaction costs.
\newblock \emph{Finance and stochastics}, 3\penalty0 (1):\penalty0 35--54,
  1999.

\bibitem[Delbaen and Schachermayer(1994)]{delbaenschachermayer94}
Freddy Delbaen and Walter Schachermayer.
\newblock A general version of the fundamental theorem of asset pricing.
\newblock \emph{Mathematische Annalen}, 300\penalty0 (1):\penalty0 463--520,
  1994.

\bibitem[Dellacherie and Meyer(1982)]{dellacherieb}
Claude Dellacherie and Paul-Andr{\'e} Meyer.
\newblock \emph{Probabilities and Potential B}, volume~72 of
  \emph{North-Holland Mathematics Studies}.
\newblock North-Holland Publishing Co., Amsterdam, 1982.

\bibitem[El~Karoui and Quenez(1995)]{quenez}
Nicole El~Karoui and Marie-Claire Quenez.
\newblock Dynamic programming and pricing of contingent claims in an incomplete
  market.
\newblock \emph{SIAM Journal on Control and Optimization}, 33\penalty0
  (1):\penalty0 29--66, 1995.

\bibitem[F{\"o}llmer and Schied(2011)]{follmerschied}
Hans F{\"o}llmer and Alexander Schied.
\newblock \emph{Stochastic finance: an introduction in discrete time}.
\newblock Walter de Gruyter, 2011.

\bibitem[Guasoni et~al.(2008)Guasoni, R{\'a}sonyi, and
  Schachermayer]{schachermayercps}
Paolo Guasoni, Mikl{\'o}s R{\'a}sonyi, and Walter Schachermayer.
\newblock Consistent price systems and face-lifting pricing under transaction
  costs.
\newblock \emph{The Annals of Applied Probability}, pages 491--520, 2008.

\bibitem[Guasoni et~al.(2010)Guasoni, R{\'a}sonyi, and
  Schachermayer]{guasoniftap}
Paolo Guasoni, Mikl{\'o}s R{\'a}sonyi, and Walter Schachermayer.
\newblock The fundamental theorem of asset pricing for continuous processes
  under small transaction costs.
\newblock \emph{Annals of Finance}, 6\penalty0 (2):\penalty0 157--191, 2010.

\bibitem[Kabanov(1999)]{kabanov99}
Yuri~M. Kabanov.
\newblock Hedging and liquidation under transaction costs in currency markets.
\newblock \emph{Finance and Stochastics}, 3\penalty0 (2):\penalty0 237--248,
  1999.

\bibitem[Kabanov and Last(2002)]{kabanovlast}
Yuri~M. Kabanov and G{\"u}nter Last.
\newblock Hedging under transaction costs in currency markets: a
  continuous-time model.
\newblock \emph{Mathematical Finance}, 12\penalty0 (1):\penalty0 63--70, 2002.

\bibitem[Kabanov and Safarian(2009)]{kabanovbook}
Yuri~M. Kabanov and Mher Safarian.
\newblock \emph{Markets with transaction costs: Mathematical Theory}.
\newblock Springer Science \& Business Media, 2009.

\bibitem[Klenke(2013)]{klenke2013probability}
Achim Klenke.
\newblock \emph{Probability theory: a comprehensive course}.
\newblock Springer Science \& Business Media, 2013.

\bibitem[Kramkov(1996)]{kramkovduality}
Dmitrij~O Kramkov.
\newblock Optional decomposition of supermartingales and hedging contingent
  claims in incomplete security markets.
\newblock \emph{Probability Theory and Related Fields}, 105\penalty0
  (4):\penalty0 459--479, 1996.

\bibitem[Lo{\`e}ve(1978)]{loeve1978probability}
Michel Lo{\`e}ve.
\newblock \emph{Probability Theory II}.
\newblock F.W.Gehring P.r.Halmos and C.c.Moore. Springer, 1978.

\bibitem[Nutz(2015)]{nutz2015robust}
Marcel Nutz.
\newblock Robust superhedging with jumps and diffusion.
\newblock \emph{Stochastic Processes and their Applications}, 125\penalty0
  (12):\penalty0 4543--4555, 2015.

\bibitem[Nutz and Soner(2012)]{nutz2012superhedging}
Marcel Nutz and H~Mete Soner.
\newblock Superhedging and dynamic risk measures under volatility uncertainty.
\newblock \emph{SIAM Journal on Control and Optimization}, 50\penalty0
  (4):\penalty0 2065--2089, 2012.

\bibitem[Ob{\l}{\'o}j and Wiesel(2021)]{obloj2021robust}
Jan Ob{\l}{\'o}j and Johannes Wiesel.
\newblock Robust estimation of superhedging prices.
\newblock \emph{The Annals of Statistics}, 49\penalty0 (1):\penalty0 508--530,
  2021.

\bibitem[Protter(2005)]{protterstochastic}
Philip~E Protter.
\newblock Stochastic differential equations.
\newblock In \emph{Stochastic integration and differential equations}, pages
  249--361. Springer, 2005.

\bibitem[Reitsam(2021)]{thesis}
Thomas Reitsam.
\newblock \emph{Asset price bubbles and dynamic super-replication under
  transaction costs}.
\newblock PhD thesis, LMU Munich, 2021.

\bibitem[Schachermayer(2014{\natexlab{a}})]{schachermayeradmissible}
Walter Schachermayer.
\newblock Admissible trading strategies under transaction costs.
\newblock In \emph{S{\'e}minaire de Probabilit{\'e}s XLVI}, pages 317--331.
  Springer, 2014{\natexlab{a}}.

\bibitem[Schachermayer(2014{\natexlab{b}})]{schachermayersuperreplication}
Walter Schachermayer.
\newblock The super-replication theorem under proportional transaction costs
  revisited.
\newblock \emph{Mathematics and Financial Economics}, 8\penalty0 (4):\penalty0
  383--398, 2014{\natexlab{b}}.

\bibitem[Touzi(1999)]{touzi1999super}
Nizar Touzi.
\newblock Super-replication under proportional transaction costs: from discrete
  to continuous-time models.
\newblock \emph{Mathematical methods of operations research}, 50\penalty0
  (2):\penalty0 297--320, 1999.

\end{thebibliography}
\bibliographystyle{plainnat}

\end{document}